\documentclass[journal]{IEEEtran}

\IEEEoverridecommandlockouts
\usepackage[dvips]{graphicx}
\usepackage{algorithm}
\usepackage{algorithmic}
\usepackage{float}
\usepackage[ancient]{flushend} 
\usepackage{psfrag}
\newcommand{\norm}[1]{\left\lVert#1\right\rVert}
\usepackage{bm}
\usepackage{amssymb,mathtools}
\usepackage{amsfonts}
\usepackage{cite}
\usepackage{amsmath,amssymb}
\usepackage{amsthm} 
\usepackage{subfigure}
\newtheorem{theo}{Theorem}
\newtheorem{lem}{Lemma}
\newtheorem{cor}{Corollary}
\theoremstyle{definition}
\newtheorem{defi}{Definition}

\ifCLASSINFOpdf

\else

\fi

\begin{document}

\title{Privacy-Preserving Distributed Zeroth-Order
Optimization}

\author{Cristiano~Gratton,~\IEEEmembership{Member,~IEEE,}
        Naveen~K.~D.~Venkategowda,~\IEEEmembership{Member,~IEEE,}
				Reza~Arablouei,
        and~Stefan~Werner,~\IEEEmembership{Senior~Member,~IEEE}
\thanks{This work was partly supported by the Research Council of Noway. A conference precursor of this work appears in the \emph{Proceedings of the European Signal Processing Conference}, Amsterdam, NL, January 2021~\cite{Gratton2020eusipco}.}	
\thanks{C. Gratton, N. K. D. Venkategowda, and S. Werner are with the Department of Electronic Systems, Norwegian University of Science and Technology, Trondheim, Norway (email:cristiano.gratton@ntnu.no; naveen.dv@ntnu.no; stefan.werner@ntnu.no).}
\thanks{R. Arablouei is with the Commonwealth Scientific and Industrial Research Organisation, Pullenvale QLD 4069, Australia (email:reza.arablouei@csiro.au).}}


\maketitle

\begin{abstract}
We develop a privacy-preserving distributed algorithm to minimize a regularized empirical risk function when the first-order information is not available and data is distributed over a multi-agent network. We employ a zeroth-order method to minimize the associated augmented Lagrangian function in the primal domain using the alternating direction method of multipliers (ADMM). We show that the proposed algorithm, named distributed zeroth-order ADMM (D-ZOA), has intrinsic privacy-preserving properties. Unlike the existing privacy-preserving methods based on the ADMM where the primal or the dual variables are perturbed with noise, the inherent randomness due to the use of a zeroth-order method endows D-ZOA with intrinsic differential privacy. By analyzing the perturbation of the primal variable, we show that the privacy leakage of the proposed D-ZOA algorithm is bounded. In addition, we employ the moments accountant method to show that the total privacy leakage grows sublinearly with the number of ADMM iterations. D-ZOA outperforms the existing differentially private approaches in terms of accuracy while yielding the same privacy guarantee. We prove that D-ZOA converges to the optimal solution at a rate of $\mathcal{O}(1/M)$ where $M$ is the number of ADMM iterations. The convergence analysis also reveals a practically important trade-off between privacy and accuracy. Simulation results verify the desirable privacy-preserving properties of D-ZOA and its superiority over a state-of-the-art algorithm as well as its network-wide convergence to the optimal solution.
\end{abstract}

\begin{IEEEkeywords}
Alternating direction method of multipliers, differential privacy, distributed optimization, zeroth-order optimization methods.
\end{IEEEkeywords}

\IEEEpeerreviewmaketitle

\section{Introduction}

\IEEEPARstart{W}{ith} the recent advances in technology, large amounts of data are gathered by numerous sensors scattered over large geographical areas. Performing learning tasks at a central processing hub in a large distributed network can be prohibitive due to computation/communication costs. Collecting all data at a central hub may also create a single point of failure. Therefore, it is important to develop algorithms that are capable of processing the data gathered by agents dispersed over a distributed network \cite{Mateos2010, Gratton2019, Grattonasilomar2018, Akhtar2018, Venkategowdaglobalsip2018, Giannakis2016, Hajinezhad2019, Nedic2009, Talebi2019}.
In this context, each agent has access only to the information of its local objective function while the agents aim to collaboratively optimize the aggregate of the local objective functions. Such distributed solutions are highly demanded in many of today's optimization problems pertaining to statistics \cite{Mateos2010,Gratton2019,Grattonasilomar2018}, signal processing \cite{Akhtar2018, Venkategowdaglobalsip2018, Giannakis2016}, and control \cite{Hajinezhad2019, Nedic2009,Talebi2019}.

However, the communications between neighboring agents may lead to privacy violation issues. An adversary may infer sensitive data of one or more agents by sniffing the communicated information. The adversary can be either a curious member of the network or an eavesdropper. Therefore, it is important to develop privacy-preserving methods that allow  distributed processing of data without revealing private information. Differential privacy provides privacy protection against adversarial attacks by ensuring minimal change in the outcome of the algorithm regardless of whether or not a single agent's data is taken into account.

Moreover, in some real-world problems, obtaining first-order information is hard due to non-smooth objectives \cite{Nedic2009, Mateos2010, Hajinezhad2019} or lack of any complete objective function. For example, in bandit optimization \cite{Agarwal2010}, an adversary generates a sequence of loss functions and the goal is to minimize such sequence that is only available at some agents. In addition, in simulation-based optimization, the objective is available only using repeated simulations \cite{Spall2003} while, in adversarial black-box machine learning models, only the function values are given \cite{Chen2017}. 
This motivates the use of zeroth-order methods, which only use the values of the objective functions to approximate their gradients.

\subsection{Related Work}

There have been several works developing privacy-preserving algorithms for distributed convex optimization \cite{Tao2017,Xueruallerton,Xueru2018,Zhuhan,Zhuhan2019,Mitra2015,Pappas,Hale}. The work in \cite{Tao2017} proposes two differentially private distributed algorithms that are based on the alternating direction method of multipliers (ADMM). The algorithms in \cite{Tao2017} are obtained by perturbing the dual and the primal variable, respectively. However, in both algorithms, the privacy leakage of an agent is bounded only at a single iteration and an adversary might exploit knowledge available from all iterations to infer sensitive information. This shortcoming is mitigated in \cite{Xueruallerton,Xueru2018,Zhuhan,Zhuhan2019}. The works in \cite{Xueruallerton,Xueru2018} develop ADMM-based differentially private algorithms with improved accuracy. The work in \cite{Zhuhan} employs the ADMM to develop a distributed algorithm where the primal variable is perturbed by adding a Gaussian noise with diminishing variance to ensure zero-concentrated differential privacy enabling higher accuracy compared to the common $(\epsilon,\delta)$-differential privacy. The work in \cite{Zhuhan2019} develops a stochastic ADMM-based distributed algorithm that further enhances the accuracy while ensuring differential privacy. The authors of \cite{Mitra2015,Pappas,Hale} propose differentially-private distributed algorithms that utilize the projected-gradient-descent method for handling constraints. The differentially private distributed algorithm proposed in \cite{Cortes} is based on perturbing the local objective functions. 

All the above-mentioned algorithms offer distributed solutions only for problems with smooth objective functions. The work in \cite{Huang2020} addresses problems with non-smooth objective functions by employing a first-order approximation of the augmented Lagrangian with a scalar $l_2$-norm proximity operator. However, this algorithm is not fully distributed since it requires a central coordinator to average all the perturbed primal variable updates over the network at every iteration.

Most existing algorithms require some modifications through deliberately perturbing either the local estimates or the objective functions. This compromises the performance of the algorithm by degrading its accuracy especially when large amount of noise is required to provide high privacy levels. The work in \cite{intrinsic} considers privacy-preserving properties that are intrinsic, i.e., they do not require any change in the algorithm but are associated with the algorithm's inherent properties. However, the approach taken in \cite{intrinsic} considers a privacy metric based on the topology of the communication graph.

\subsection{Contributions}

In this paper, we develop a fully-distributed differentially-private algorithm to solve a class of regularized empirical risk minimization (ERM) problems when first-order information is unavailable or hard to obtain. We utilize the ADMM for distributed optimization and a zeroth-order method, called the two-point stochastic gradient algorithm \cite{Duchi2015}, to minimize the augmented Lagrangian function in the ADMM's primal update step. The proposed algorithm, called distributed zeroth-order ADMM (D-ZOA), is fully distributed in the sense that each agent of the network communicates only with its immediate neighbors and no central coordination is necessary. It also only requires the objective function values to solve the underlying ERM problem while respecting privacy.

The privacy-preserving properties of the proposed D-ZOA algorithm are intrinsic. To prove this, we model the primal variable at each agent as the sum of an exact (unperturbed) value and a random perturbation. This enables us to approximate the distribution of the primal variable and verify that the stochasticity inherent to the employed zeroth-order method makes D-ZOA differential private. To this end, we prove that the privacy leakage of a single iteration of D-ZOA at each agent is bounded. Utilizing the moments accountant method \cite{Abadi}, we also show that the total privacy leakage over all iterations grows sublinearly with the number of ADMM iterations. This is particularly significant as we observe that the optimization accuracy of D-ZOA is higher compared to the existing privacy-preserving approaches, which perturb the variables exchanged among the network agents by adding noise. 

We prove that the convergence rate of D-ZOA is $\mathcal{O}(1/M)$ where $M$ is the number of iterations of the ADMM outer loop. The convergence rate of the inner loop is $\mathcal{O}(\sqrt{P/T})$ where $P$ is the number of features and $T$ is the number of iterations of the inner loop. More importantly, no communication among agents is required throughout the inner loop. The convergence analysis also reveals an explicit privacy-accuracy trade-off where a stronger privacy guarantee corresponds to a lower accuracy.

Simulation results demonstrate that, with any given level of required privacy guarantee, D-ZOA outperforms an existing ADMM-based related algorithm (DP-ADMM) presented in \cite{Huang2020}. This algorithm has been developed for distributed optimization of non-smooth objective functions and achieves differential privacy by adding noise to the primal variables~\cite{Huang2020}.

\subsection{Paper Organization}

The rest of the paper is organized as follows. In Section \ref{secttwo}, we describe the system model and formulate the distributed ERM problem when first-order information is not available. In Section \ref{sectthree}, we describe our proposed D-ZOA algorithm and explain the privacy issues associated with distributed learning. In Section \ref{sectfour}, we present the intrinsic privacy-preserving properties of the proposed D-ZOA algorithm by showing that the privacy leakage of each agent at any iteration is bounded and the total privacy leakage grows sublinearly with the number of ADMM iterations. In Section \ref{sectfive}, we prove the convergence of D-ZOA by confirming that both inner and outer loops of the algorithm converge. We provide some simulation results in Section \ref{sectsix} and draw conclusions in Section \ref{sectseven}.

\subsection{Mathematical notations}

The set of natural and real numbers are denoted by $\mathbb{N}$ and $\mathbb{R}$, respectively. The set of positive real numbers is denoted by $\mathbb{R}_+$. Scalars, column vectors, and matrices are respectively denoted by lowercase, bold lowercase, and bold uppercase letters. The operators $(\cdot)^\mathsf{T}$, $\text{det}(\cdot)$, and $\text{tr}(\cdot)$ denote transpose, determinant, and trace of a matrix, respectively. $\norm{\cdot}$ represents the Euclidean norm of its vector argument. $\mathbf{I}_n$ is an identity matrix of size $n$, $\mathbf{0}_n$ is an $n\times 1$ vector with all zeros entries, $\mathbf{0}_{n\times p}=\mathbf{0}_n\mathbf{0}_p^\mathsf{T}$, and $|\cdot|$ denotes the cardinality if its argument is a set. The statistical expectation and covariance operators are represented by $\mathbb{E}[\cdot]$ and $\text{cov}[\cdot]$, respectively. The notation  $\mathcal{N}(\boldsymbol{\mu},\boldsymbol{\Sigma})$ denotes normal distribution with mean $\boldsymbol{\mu}$ and covariance matrix $\boldsymbol{\Sigma}$.  For a positive semidefinite matrix $\mathbf{X}$, $\lambda_{\min}(\mathbf{X})$ and $\lambda_{\max}(\mathbf{X})$ denote the nonzero smallest and largest eigenvalues of $\mathbf{X}$, respectively. For a vector $\mathbf{x}\in\mathbb{R}^n$ and a matrix $\mathbf{A}$, $\norm{\mathbf{x}}^2_{\mathbf{A}}$ denotes the quadratic form $\mathbf{x}^\mathsf{T}\mathbf{A}\mathbf{x}$.

\section{System Model} \label{secttwo}

We consider a network with $K\in\mathbb{N}$ agents and $E\in\mathbb{N}$ edges modeled as an undirected graph $\mathcal{G}(\mathcal{K},\mathcal{E})$ where the vertex set $\mathcal{K}=\{1,\dots,K\}$ corresponds to the agents and the set $\mathcal{E}$ represents the bidirectional communication links between the pairs of neighboring agents. Agent $k\in\mathcal{K}$ can communicate only with the agents in its neighborhood $\mathcal{N}_k$. By convention, the set $\mathcal{N}_k$ includes the agent $k$ as well. 

Each agent $k\in\mathcal{K}$ has a private dataset
\begin{align*}
\mathcal{D}_k=\{(\mathbf{X}_k,\mathbf{y}_k):
\mathbf{X}_k&=[\mathbf{x}_{k,1},\mathbf{x}_{k,2},\hdots,\mathbf{x}_{k,N_k}]^\mathsf{T}\in\mathbb{R}^{N_k\times P},\\
\mathbf{y}_k&=[y_{k,1},y_{k,2},\hdots,y_{k,N_k}]^\mathsf{T}\in\mathbb{R}^{N_k}\}
\end{align*}
where $N_k$ is the number of data samples collected at the agent $k$ and $P$ is the number of features in each sample. 

We consider the problem of estimating a parameter of interest $\boldsymbol{\beta}\in\mathbb{R}^{P}$ that relates the value of an output measurement stored in the response vector $\mathbf{y}_k$ to input measurements collected in the corresponding row of the local matrix $\mathbf{X}_k$. The associated supervised learning problem can be cast as a regularized ERM expressed by
\begin{equation}
\label{firsteq}
\min_{\boldsymbol{\beta}}\sum_{k=1}^K\frac{1}{N_k}\sum_{j=1}^{N_k}\ell(\mathbf{x}_{k,j},y_{k,j};\boldsymbol{\beta})+\eta R(\boldsymbol{\beta})
\end{equation}
where $\ell:\mathbb{R}^{P}\rightarrow\mathbb{R}$ is the loss function, $R:\mathbb{R}^{P}\rightarrow\mathbb{R}$ is the regularizer function, and $\eta>0$ is the regularization parameter. The ERM problem pertains to several applications in machine learning, e.g., linear regression \cite{Mateos2010}, support vector machine \cite{Forero2010}, and logistic regression \cite{Huang2020,Zhuhan}. We assume that the loss function $\ell(\cdot)$ and the regularizer function $R(\cdot)$ are both convex but at least one of them is \textit{non-smooth}. Let us denote the optimal solution of \eqref{firsteq} by $\boldsymbol{\beta}^c$.

\section{Non-Smooth Distributed Learning} \label{sectthree}

We first discuss the consensus-based reformulation of the problem that allows its distributed solution through an iterative process consisting of two nested loops. Then, we describe the ADMM procedure that forms the outer loop and the zeroth-order two-point stochastic gradient algorithm that constitutes the inner loop solving the ADMM primal update step. Finally, we discuss the related privacy matters.

\subsection{Consensus-Based Reformulation}

To solve \eqref{firsteq} in a distributed manner, we reformulate it as the following constrained minimization problem 

\begin{equation}
\begin{aligned}
&\underset{\{\boldsymbol{\beta}_k\}}{\min} 
&& \sum_{k=1}^{K}\Bigl(\frac{1}{N_k}\sum_{j=1}^{N_k}\ell(\mathbf{x}_{k,j},y_{k,j};\boldsymbol{\beta}_k)+\frac{\eta}{K}R(\boldsymbol{\beta}_k)\Bigr)  \\
&\text{\ s.t.}\ 
&& \boldsymbol{\beta}_k=\boldsymbol{\beta}_l, \quad l\in\mathcal{N}_k, \quad \forall k\in\mathcal{K}
\end{aligned}
\label{secondeq}
\end{equation}
where $\mathcal{V}=\{\boldsymbol{\beta}_k\}_{k=1}^K$ are the primal variables representing local copies of $\boldsymbol{\beta}$ at the agents. The equality constraints impose consensus across each agent's neighborhood $\mathcal{N}_k$. To solve \eqref{secondeq} collaboratively and in a fully-distributed manner, we utilize the ADMM \cite{Giannakis2016}. For this purpose, we rewrite \eqref{secondeq} as
\begin{equation}
\begin{aligned}
&\underset{\{\boldsymbol{\beta}_k\}}{\min}
&& \sum_{k=1}^{K}\Bigl(\frac{1}{N_k}\sum_{j=1}^{N_k}\ell(\mathbf{x}_{k,j},y_{k,j};\boldsymbol{\beta}_k)+\frac{\eta}{K}R(\boldsymbol{\beta}_k)\Bigr) \\
&\text{\ s.t.}\ 
&& \boldsymbol{\beta}_k=\mathbf{z}_k^l,\ \boldsymbol{\beta}_l=\mathbf{z}_k^l, \quad l\in\mathcal{N}_k, \quad\forall k\in\mathcal{K}
\end{aligned}
\label{thirdeq}
\end{equation}
where $\mathcal{Z}=\{\mathbf{z}_k^l\}_{k \in\mathcal{K}, l\in\mathcal{N}_k}$ are the auxiliary variables yielding an alternative but equivalent representation of the constraints in \eqref{secondeq}. They help decouple $\boldsymbol{\beta}_k$ in the constraints and facilitate the derivation of the local recursions before being eventually eliminated.

To apply the ADMM, we rewrite \eqref{thirdeq} in the matrix form. By defining $\mathbf{w}\in\mathbb{R}^{KP}$ concatenating all $\boldsymbol{\beta}_k$ and $\mathbf{z}\in\mathbb{R}^{2EP}$ concatenating all $\mathbf{z}_k^l$, \eqref{thirdeq} can be written as  
\begin{equation}
\begin{aligned}
&\underset{\mathbf{w},\mathbf{z}}{\min}
&& f(\mathbf{w}) \\
&\text{\ s.t.}\ 
&& \mathbf{A}\mathbf{w}+\mathbf{B}\mathbf{z}=\mathbf{0}
\end{aligned}
\label{convergence3}
\end{equation}
where
\begin{equation*}
\begin{aligned}
f(\mathbf{w})&=\sum_{k=1}^Kf_k(\boldsymbol{\beta}_k),\\
f_k(\boldsymbol{\beta}_k)&=\frac{1}{N_k}\sum_{j=1}^{N_k}\ell(\mathbf{x}_{k,j},y_{k,j};\boldsymbol{\beta}_k)+\frac{\eta}{K}R(\boldsymbol{\beta}_k),
\end{aligned}
\end{equation*}
$\mathbf{A}=[\mathbf{A}_1^\mathsf{T},\mathbf{A}_2^\mathsf{T}]^\mathsf{T}$, and $\mathbf{A}_1,\mathbf{A}_2\in\mathbb{R}^{2EP\times KP}$ are both composed of $2E\times K$ blocks of $P\times P$ matrices. If $(k,l)\in\mathcal{E}$ and $\mathbf{z}_k^l$ is the $q$th block of $\mathbf{z}$, then the $(q,k)$th block of $\mathbf{A}_1$ and the $(q,l)$th block of $\mathbf{A}_2$ are the identity matrix $\mathbf{I}_P$. Otherwise, the corresponding blocks are $\mathbf{0}_{P\times P}$. Furthermore, we have
\begin{equation*}
\mathbf{B}=[-\mathbf{I}_{2EP},-\mathbf{I}_{2EP}]^\mathsf{T}.
\end{equation*}
To facilitate the representation, we also define the following matrices
\begin{align*}
\mathbf{M}_+&=\mathbf{A}_1^\mathsf{T}+\mathbf{A}_2^\mathsf{T}\\ \mathbf{M}_-&=\mathbf{A}_1^\mathsf{T}-\mathbf{A}_2^\mathsf{T}\\ \mathbf{L}_+&=0.5\mathbf{M}_+\mathbf{M}_+^\mathsf{T}\\ \mathbf{L}_-&=0.5\mathbf{M}_-\mathbf{M}_-^\mathsf{T}\\
\mathbf{H}&=0.5(\mathbf{L}_++\mathbf{L}_-)\\
\mathbf{Q}&=\sqrt{0.5\mathbf{L}_-}.
\end{align*}

Solving \eqref{convergence3} via the ADMM requires a procedure that is described in the next subsection.

\subsection{Distributed ADMM Algorithm}

To solve the minimization problem \eqref{convergence3} in a distributed manner, we employ the ADMM \cite{Wotao2014}. The augmented Lagrangian function associated with \eqref{convergence3} is given by
\begin{equation}
\label{eqfourth}
\mathcal{L}_{\rho}(\mathbf{w},\mathbf{z},\boldsymbol{\lambda})=f(\mathbf{w})+\boldsymbol{\lambda}^\mathsf{T}(\mathbf{A}\mathbf{w}+\mathbf{B}\mathbf{z})+\frac{\rho}{2}\norm{\mathbf{A}\mathbf{w}+\mathbf{B}\mathbf{z}}^2
\end{equation}
where $\boldsymbol{\lambda}\in\mathbb{R}^{4EP}$ is the Lagrange multiplier and $\rho>0$ is a penalty parameter.

The ADMM entails an iterative procedure consisting of three steps at each iteration $m$ as
\begin{equation}
\begin{aligned}
\mathbf{w}^{(m)}&=\arg\min_{\mathbf{w}}\mathcal{L}_{\rho}(\mathbf{w},\mathbf{z}^{(m-1)},\boldsymbol{\lambda}^{(m-1)})\\
\mathbf{z}^{(m)}&=\arg\min_{\mathbf{z}}\mathcal{L}_{\rho}(\mathbf{w}^{(m)},\mathbf{z},\boldsymbol{\lambda}^{(m-1)})\\
\boldsymbol{\lambda}^{(m)}&=\boldsymbol{\lambda}^{(m-1)}+\rho[\mathbf{A}\mathbf{w}^{(m)}+\mathbf{B}\mathbf{z}^{(m)}].
\end{aligned}
\label{eqfifth}
\end{equation}
In \cite{Wotao2014}, it is shown that, considering $\boldsymbol{\lambda}=[\boldsymbol{\lambda}_1^\mathsf{T},\boldsymbol{\lambda}_2^\mathsf{T}]^\mathsf{T}$ with $\boldsymbol{\lambda}_1,\boldsymbol{\lambda}_2\in\mathbb{R}^{2EP}$ and $\boldsymbol{\gamma}=\mathbf{M}_-\boldsymbol{\lambda}_1$, the ADMM algorithm steps in \eqref{eqfifth} reduce to
\begin{align}
&\mathbf{w}^{(m)}=\arg\min_{\mathbf{w}}\mathcal{F}(\mathbf{w},\mathbf{w}^{(m-1)},\boldsymbol{\gamma}^{(m-1)}) \label{eqsixtha}\\
&\boldsymbol{\gamma}^{(m)}=\boldsymbol{\gamma}^{(m-1)}+\rho\mathbf{L}_-\mathbf{w}^{(m)} \label{eqsixthb}
\end{align}
where
\begin{align*}
\mathcal{F}(\mathbf{w},\mathbf{w}^{(m-1)},\boldsymbol{\gamma}^{(m-1)})&= f(\mathbf{w})+\mathbf{w}^\mathsf{T}\boldsymbol{\gamma}^{(m-1)}\\
&+\rho\mathbf{w}^\mathsf{T}\mathbf{H}\mathbf{w}-\rho\mathbf{w}^\mathsf{T}\mathbf{L}_+\mathbf{w}^{(m-1)}
\end{align*}
and the initial values of $\mathbf{w}$ and $\boldsymbol{\gamma}$ are set to zero. Note that the update equations \eqref{eqsixtha} and \eqref{eqsixthb} are distributed among the agents as $\mathbf{w}=[\boldsymbol{\beta}_1^\mathsf{T},\boldsymbol{\beta}_2^\mathsf{T},\hdots,\boldsymbol{\beta}_K^\mathsf{T}]^\mathsf{T}$ and $\boldsymbol{\gamma}=[\boldsymbol{\gamma}_1^\mathsf{T},\boldsymbol{\gamma}_2^\mathsf{T},\hdots,\boldsymbol{\gamma}_K^\mathsf{T}]^\mathsf{T}$ where
\begin{equation*}
\boldsymbol{\gamma}_k=2\sum_{l\in\mathcal{N}_k}\boldsymbol{\gamma}_k^l
\end{equation*}
is the local Lagrange multiplier at agent $k$ associated with the constraints in \eqref{thirdeq} \cite{Giannakis2016}.

Since the objective function in \eqref{eqsixtha} is assumed to be non-smooth, it cannot be solved using any first-order method. To overcome this, we use a zeroth-order method described in the next subsection.

\subsection{Zeroth-Order Method}

To solve \eqref{eqsixtha} employing a zeroth-order method, we make the following assumptions that are common in the zeroth-order optimization literature, see, e.g., \cite{Hajinezhad2019,Duchi2015,Nesterov2017}.

\textit{Assumption 1:} 
The function $\mathcal{F}(\cdot)$ is closed.

\textit{Assumption 2:} 
The function $\mathcal{F}(\cdot)$ is Lipschitz-continuous with the Lipschitz constant $L$.

In zeroth-order optimization, $\mathcal{F}(\cdot)$ is also required to be convex \cite{Hajinezhad2019,Duchi2015,Nesterov2017}. However, $\mathcal{F}(\cdot)$ is the sum of $f(\cdot)$, which are assumed to be convex, the linear terms $\mathbf{w}^\mathsf{T}\boldsymbol{\gamma}^{(m-1)}$ and $-\rho\mathbf{w}^\mathsf{T}\mathbf{L}_+\mathbf{w}^{(m-1)}$, and the quadratic term $\rho\mathbf{w}^\mathsf{T}\mathbf{H}\mathbf{w}$, which is convex since $\mathbf{H}$ is positive semi-definite for being the sum of the positive semi-definite matrices $\mathbf{L_+}$ and $\mathbf{L_-}$. Therefore, $\mathcal{F}(\cdot)$ is convex as it is the sum of convex functions \cite{BoydStephenP2004Co}.

We utilize the two-point stochastic-gradient algorithm that has been proposed in \cite{Duchi2015} for optimizing general non-smooth functions. More specifically, we use the stochastic mirror descent method with the proximal function $\frac{1}{2}\norm{\cdot}$ and the gradient estimator at point $\mathbf{w}$ given by 
\begin{align}
&\Gamma(\mathbf{w},\boldsymbol{\gamma}^{(m-1)},u_1,u_2,\boldsymbol{\nu}_1,\boldsymbol{\nu}_2) =  u_2^{-1}[\mathcal{F}(\mathbf{w}+u_1\boldsymbol{\nu}_1\nonumber\\
&\qquad\quad+u_2\boldsymbol{\nu}_2,\boldsymbol{\gamma}^{(m-1)})-\mathcal{F}(\mathbf{w}+u_1\boldsymbol{\nu}_1,\boldsymbol{\gamma}^{(m-1)})]\boldsymbol{\nu}_2
\label{firstzeroth}
\end{align}
where $u_1>0$ and $u_2>0$ are smoothing constants and $\boldsymbol{\nu}_1$, $\boldsymbol{\nu}_2$ are independent zero-mean Gaussian random vectors with the covariance matrix $\mathbf{I}_{KP}$, i.e., 
$\boldsymbol{\nu}_1,\boldsymbol{\nu}_2\sim\mathcal{N}(\mathbf{0}_{KP},\mathbf{I}_{KP})$.

The two-point stochastic-gradient algorithm consists of two randomization steps where the second step is aimed at preventing the perturbation vector $\boldsymbol{\nu}_2$ from being close to a point of non-smoothness \cite{Duchi2015}. This algorithm entails an iterative procedure that consists of three steps at each iteration $t$. First, $J\in\mathbb{N}$ independent random vectors $\{\boldsymbol{\nu}_{1,t}^j\}_{j=1}^J$ and $\{\boldsymbol{\nu}_{2,t}^j\}_{j=1}^J$ are sampled from $\mathcal{N}(\mathbf{0}_{KP},\mathbf{I}_{KP})$. Second, a stochastic gradient $\mathbf{g}^{(t)}$ is calculated as
\begin{equation}
\label{secondzeroth}
\mathbf{g}^{(t)} = \frac{1}{J}\sum_{j=1}^J\mathbf{g}_j^{(t)}
\end{equation}
where
\begin{align*}
\mathbf{g}_j^{(t)}=\Gamma(\mathbf{w}^{(t)},\boldsymbol{\gamma}^{(m-1)},u_{1,t},u_{2,t},\boldsymbol{\nu}_{1,t}^j,\boldsymbol{\nu}_{2,t}^j),
\end{align*}
$\mathbf{w}^{(t)}$ is the $t$th iterate of the two-point stochastic-gradient algorithm with the initial value $\mathbf{w}^{(0)}=\mathbf{0}$ and $\{u_{1,t}\}_{t=1}^{\infty}$ and $\{u_{2,t}\}_{t=1}^{\infty}$ are two non-increasing sequences of positive parameters such that $u_{2,t}\le u_{1,t}/2$. Finally, $\mathbf{w}^{(t)}$ is updated as 
\begin{equation}
\label{thirdzeroth}
\mathbf{w}^{(t)}=\mathbf{w}^{(t-1)}-\alpha_{t}\mathbf{g}^{(t)}
\end{equation}
where $\alpha_{t}$ is a time-varying step-size. The step-size is computed as
\begin{equation*}
\alpha_{t} = \left(L\sqrt{tP\log(2P)}\right)^{-1}\alpha_{0} R    
\end{equation*}
where $\alpha_{0}$ is an appropriate initial step-size and $R$ is an upper bound on the distance between a minimizer $\mathbf{w}^*$ to \eqref{eqsixtha} and the first iterate $\mathbf{w}^{(1)}$ as per \cite{Duchi2015}. 

Locally at agent $k$, the appropriate subvectors $\{\boldsymbol{\nu}_{1,t}^{j,k}\}_{j=1}^J$ and $\{\boldsymbol{\nu}_{2,t}^{j,k}\}_{j=1}^J$ of, respectively, $\{\boldsymbol{\nu}_{1,t}^j\}_{j=1}^J$ and $\{\boldsymbol{\nu}_{2,t}^j\}_{j=1}^J$ are utilized. These subvectors are also independent of each other and have independent entries. Subsequently, a $k$-local stochastic gradient $\mathbf{g}_k^{(t)}$ is computed as 
\begin{equation}
\label{fourthzeroth}
\mathbf{g}_k^{(t)} = \frac{1}{J}\sum_{j=1}^J\mathbf{g}_{j,k}^{(t)}
\end{equation}
where
\begin{equation*}
\mathbf{g}_{j,k}^{(t)}=\Gamma(\boldsymbol{\beta}_k^{(t)},\boldsymbol{\gamma}_k^{(m-1)},u_{1,t},u_{2,t},\boldsymbol{\nu}_{1,t}^{j,k},\boldsymbol{\nu}_{2,t}^{j,k}).
\end{equation*}
Hence, $\boldsymbol{\beta}_k^{(t)}$ is updated as
\begin{equation}
\label{fifthzeroth}
\boldsymbol{\beta}_k^{(t)}=\boldsymbol{\beta}_k^{(t-1)}-\alpha_{t}\mathbf{g}_k^{(t)}.
\end{equation}
We use multiple independent random samples $\{\boldsymbol{\nu}_{1,t}^j\}_{j=1}^J$ and $\{\boldsymbol{\nu}_{2,t}^j\}_{j=1}^J$ to obtain a more accurate estimate of the gradient $\mathbf{g}^{(t)}$ as remarked in \cite{Duchi2015}. 

Furthermore, no communication among agents is needed in the inner loop. The update equations in \eqref{thirdzeroth} and \eqref{fifthzeroth} can be implemented in a fully-distributed fashion since they involve only the variables available within every agent's neighborhood.
The proposed algorithm, D-ZOA, is summarized in Algorithm~1.

\begin{algorithm}[t!]
\caption{Distributed Zeroth-Order ADMM (D-ZOA)}

\label{alg:IDP-D-ZOA}

 \begin{algorithmic}

  \STATE At all agents $k\in\mathcal{K}$, initialize $\boldsymbol{\beta}_k^{(0)}=\mathbf{0}$, $\boldsymbol{\gamma}_k^{(0)}=\mathbf{0}$, and locally run 
   \FOR{$m=1,2,\hdots, M$} 
	 \STATE Share $\boldsymbol{\beta}_k^{(m-1)}$ with neighbors in $\mathcal{N}_k$
   \STATE Update $\boldsymbol{\gamma}_k^{(m)}$ as in \eqref{fourthprivacyc}
	\STATE Initialize $\boldsymbol{\beta}_k^{(0)}=\mathbf{0}$
	\FOR{$t=1,2,\hdots,T$}
	 \STATE Draw independent $\{\boldsymbol{\nu}_{1,t}^{j,k}\}_{j=1}^J,\{\boldsymbol{\nu}_{2,t}^{j,k}\}_{j=1}^J\sim\mathcal{N}(\mathbf{0}_{P},\mathbf{I}_{P})$
	\STATE Set $u_{1,t}=u_1/t$, $u_{2,t}=u_1/(Pt)^2$
	\STATE Compute $\mathbf{g}_k^{(t)}$ as in \eqref{fourthzeroth} and \eqref{firstzeroth}
	\STATE Update $\boldsymbol{\beta}_k^{(t)}=\boldsymbol{\beta}_k^{(t-1)}-\alpha_{t}\mathbf{g}_k^{(t)}$ 
	\ENDFOR
   \STATE Update $\boldsymbol{\beta}_k^{(m)}=\boldsymbol{\beta}_k^{(T)}$
	\ENDFOR
 \end{algorithmic} 
 \end{algorithm}

\subsection{Privacy}

In Algorithm \ref{alg:IDP-D-ZOA}, the data stored at each agent, $\mathbf{X}_k$ and $\mathbf{y}_k$, is not shared with any other agent. However, the local estimates $\{\boldsymbol{\beta}_k^{(m)}\}_{k\in\mathcal{K}}$ are exchanged within the local neighborhoods. Therefore, the risk of privacy breach still exists as it has been shown by model inversion attacks \cite{fredrikson2015}. We assume that the adversary cannot access the local data, but is able to access the exchanged local estimates $\{\boldsymbol{\beta}_k^{(m)}\}_{k\in\mathcal{K}}$. The adversary can be either a member of the network or an external eavesdropper. We show that D-ZOA guarantees $(\epsilon,\delta)$-differential privacy as per below definition since it is intrinsically resistant to such inference attacks.
\begin{defi}
A randomized algorithm $\mathcal{M}$ is $(\epsilon,\delta)$-differentially private if for any two neighboring datasets $\mathcal{D}$ and $\mathcal{D}'$ differing in only one data sample and for any subset of outputs $\mathcal{O}\subseteq\text{range}(\mathcal{M})$, we have
\begin{equation}
\label{firstdef}
\text{Pr}[\mathcal{M}(\mathcal{D})\in\mathcal{O}]\le e^{\epsilon}\text{Pr}[\mathcal{M}(\mathcal{D}')\in\mathcal{O}]+\delta.
\end{equation}
This means the ratio of the probability distributions of $\mathcal{M}(\mathcal{D})$ and $\mathcal{M}(\mathcal{D}')$ is bounded by $e^{\epsilon}$.
\end{defi}

In Definition 1, $\epsilon$ and $\delta$ are privacy parameters indicating the level of privacy preservation ensured by a differentially private algorithm. A better privacy preservation is achieved with smaller $\epsilon$ or $\delta$. On the other hand, low privacy guarantee corresponds to higher values of $\epsilon$, i.e., close to $1$. Therefore, it is reasonable to assume that $\epsilon\in\mathopen(0,1\mathclose]$ as in \cite{Liu2019,Zhuhan}.

\section{Intrinsic Differential Privacy Guarantee} \label{sectfour}

Employing the zeroth-order method for the ADMM primal update produces a perturbed (inexact) estimate. Therefore, the solution in the primal update step in \eqref{eqsixtha} using D-ZOA can be modeled as 
\begin{equation}
\label{firstprivacy}
\mathbf{w}^{(m)}=\Breve{\mathbf{w}}^{(m)}+\boldsymbol{\xi}^{(m)}
\end{equation}
where $\Breve{\mathbf{w}}^{(m)}\in\mathbb{R}^P$ is the exact ADMM primal update and $\boldsymbol{\xi}^{(m)}\in\mathbb{R}^P$ is a random variable representing the perturbation. The oracle $\Breve{\mathbf{w}}^{(m)}$ satisfies the equation 
\begin{equation}
\label{secondprivacy}
\nabla f(\Breve{\mathbf{w}}^{(m)})+\boldsymbol{\gamma}^{(m-1)}+2\rho\mathbf{H}\Breve{\mathbf{w}}^{(m)}=\rho\mathbf{L}_+\mathbf{w}^{(m-1)}
\end{equation}
where $\nabla f$ is the hypothetical exact gradient of $f$. Equation \eqref{secondprivacy} is obtained by computing $\nabla \mathcal{F}(\Breve{\mathbf{w}},\mathbf{w}^{(m-1)},\boldsymbol{\gamma}^{(m-1)})$ and equating it to zero \cite[pg. 741]{Sayedbook}. The model \eqref{firstprivacy} represents an implicit primal variable perturbation that can be contrasted with the explicit primal variable perturbation used in \cite{Tao2017,Zhuhan}.

Using \eqref{firstprivacy}, the primal and dual update equations in \eqref{secondprivacy} and \eqref{eqsixthb} can be expressed as
\begin{equation}
\begin{aligned}
\Breve{\mathbf{w}}^{(m)}&=-\frac{1}{2\rho}\mathbf{H}^{-1}\nabla f(\Breve{\mathbf{w}}^{(m)})-\frac{1}{2\rho}\mathbf{H}^{-1}\boldsymbol{\gamma}^{(m-1)}\\
&+\frac{1}{2}\mathbf{H}^{-1}\mathbf{L}_+\mathbf{w}^{(m-1)}\\
\mathbf{w}^{(m)}&=\Breve{\mathbf{w}}^{(m)}+\boldsymbol{\xi}^{(m)}\\
\boldsymbol{\gamma}^{(m)}&=\boldsymbol{\gamma}^{(m-1)}+\rho\mathbf{L}_-\mathbf{w}^{(m)}.
\label{thirdprivacy}
\end{aligned}
\end{equation}

Note that \eqref{firstprivacy} also applies to all agents as
\begin{align*}
\Breve{\mathbf{w}}^{(m)}&=\left[\Breve{\boldsymbol{\beta}}_1^{(m)\mathsf{T}},\Breve{\boldsymbol{\beta}}_2^{(m)\mathsf{T}},\hdots,\Breve{\boldsymbol{\beta}}_K^{(m)\mathsf{T}}\right]^\mathsf{T}\\
\boldsymbol{\xi}^{(m)}&=\left[\boldsymbol{\xi}_1^{(m)\mathsf{T}},\boldsymbol{\xi}_2^{(m)\mathsf{T}},\hdots,\boldsymbol{\xi}_K^{(m)\mathsf{T}}\right]^\mathsf{T}
\end{align*}
where $\Breve{\boldsymbol{\beta}}_k$ is the local exact primal update at agent $k$ and $\boldsymbol{\xi}_k$ is the local perturbation of $\Breve{\boldsymbol{\beta}}_k$ at agent $k$. Recalling the definitions of $\mathbf{H}$, $\mathbf{L}_+$, and $\mathbf{L}_-$, the update equations in \eqref{thirdprivacy} entail the following local update equations at the $k$th agent: 
\begin{align}
\Breve{\boldsymbol{\beta}}_k^{(m)}=&-\frac{1}{2\rho|\mathcal{N}_k|}\nabla f_k\big(\Breve{\boldsymbol{\beta}}_k^{(m)}\big)+\frac{1}{2|\mathcal{N}_k|}\Bigl(|\mathcal{N}_k|\boldsymbol{\beta}_k^{(m-1)}\nonumber\\&+\sum_{l\in\mathcal{N}_k}\boldsymbol{\beta}_l^{(m-1)}\Bigr)-\frac{1}{2\rho|\mathcal{N}_k|}\boldsymbol{\gamma}_k^{(m-1)}\label{fourthprivacya}\\
\boldsymbol{\beta}_k^{(m)}=&\ \Breve{\boldsymbol{\beta}}_k^{(m)}+\boldsymbol{\xi}_k^{(m)}\label{fourthprivacyb}\\
\boldsymbol{\gamma}_k^{(m)}=&\ \boldsymbol{\gamma}_k^{(m-1)}+\rho\sum_{l\in\mathcal{N}_k}\Bigl(\boldsymbol{\beta}_k^{(m)}-\boldsymbol{\beta}_l^{(m)}\Bigr). \label{fourthprivacyc}
\end{align}

To prove that D-ZOA is differentially private, we require  the probability distribution of the primal variable $\boldsymbol{\beta}_k^{(m)}$. In the next subsection, we find an approximate distribution of the primal variable.

\subsection{Primal Variable Distribution}

To approximate its pdf, in view of \eqref{fifthzeroth} and the fact that $\boldsymbol{\beta}_k^{(0)}=\mathbf{0}$, we unfold $\boldsymbol{\beta}_k^{(m)}$ as
\begin{equation*}
\boldsymbol{\beta}_k^{(m)}=-\sum_{t=1}^{T}\alpha_{t}\mathbf{g}_k^{(t)}.
\end{equation*}
The stochastic gradient $\mathbf{g}_k^{(t)}$ is the average of $J$ independent random samples $\mathbf{g}_{j,k}^{(t)}$ that are functions of the random values $\{\boldsymbol{\nu}_{1,t}^j\}_{j=1}^J$ and $\{\boldsymbol{\nu}_{2,t}^j\}_{j=1}^J$ drawn from the same distribution. Therefore, it is realistic to assume that $\{\mathbf{g}_{j,k}^{(t)}\}_{j=1}^J$ are independent and identically distributed (i.i.d.) with a common mean $\boldsymbol{\mu}_k^{(t)}$ and a finite covariance matrix $\mathbf{\Psi}_k^{(t)}$. Thus, the probability distribution of $\boldsymbol{\beta}_k^{(m)}$ is given by the following lemma.
\begin{lem}
If $\{\mathbf{g}_{j,k}^{(t)}\}_{j=1}^J$ are i.i.d. and $J$ is sufficiently large, then $\boldsymbol{\beta}_k^{(m)}$ is distributed as
\begin{equation}
\label{lemmaone}
\boldsymbol{\beta}_k^{(m)}\sim\mathcal{N}\Bigl(\Breve{\boldsymbol{\beta}}_k^{(m)},\frac{1}{J}\sum_{t=1}^{T}\alpha_{t}^2\Psi_k^{(t)}\Bigr).
\end{equation}
\end{lem}
\begin{proof}\renewcommand{\qedsymbol}{}
See Appendix A.
\end{proof}

Since we are interested in exploring the privacy properties of the proposed D-ZOA and, thereby, the best privacy guarantee that it offers, we consider a conservative upper-bound for the covariance matrices $\boldsymbol{\Psi}_k^{(t)}$. This is motivated by the direct connection between the algorithmic privacy guarantee and the extent of the randomness involved in the algorithm where a higher variance corresponds to a higher privacy guarantee. Such a relationship between variance and privacy is also shown in the existing works on differential privacy, see,~e.g.,~\cite{Huang2020,Zhuhan}.\looseness=-1

To make the problem more tractable, we assume that the entries of the random vector $\boldsymbol{\beta}_k^{(m)}$ are independent of each other and have the same variance \cite{Huang2020,Tao2017,Liu2019}. Let us denote the variance of every entry of $\boldsymbol{\xi}_k^{(m)}$ by $\sigma_{k}^2$. Therefore, in view of Lemma 1, we have
\begin{equation*}
\sigma_{k}^2=\frac{1}{JP}\sum_{t=1}^{T}\alpha_t^2\text{tr}\left(\boldsymbol{\Psi}_k^{(t)}\right),
\end{equation*}
which is in turn upper-bounded as per the following lemma.

\begin{lem}
There exists a constant $c$ such that
\begin{equation}
\begin{aligned}
&\frac{1}{JP}\sum_{t=1}^{T}\alpha_t^2\text{tr}\left(\boldsymbol{\Psi}_k^{(t)}\right)\\
\le&\frac{c\alpha_0^2R^2}{JP\log(2P)}\Bigl(s_1(1+\log(P))+s_2\Bigr)-\frac{4\norm{\boldsymbol{\beta}^c}^2}{TJP}.
\label{eqfirstlemmatwo}
\end{aligned}
\end{equation}
where $s_1=\sum_{t=1}^{T}t^{-1}$,  $s_2=\sum_{t=1}^{T}t^{-1.5}$, and $\boldsymbol{\beta}^c$ is the optimal solution.
\end{lem}
\begin{proof}\renewcommand{\qedsymbol}{}
See Appendix B.
\end{proof}
In \cite{Duchi2015}, it is shown that $c=0.5$ is suitable when $\boldsymbol{\nu}_1$ and $\boldsymbol{\nu}_2$ are sampled from a multivariate normal distribution.

\subsection{$\text{l}_2$-Norm Sensitivity}

In this subsection, we estimate the $l_2$-norm sensitivity of $\Breve{\boldsymbol{\beta}}_k^{(m)}$. The $l_2$ norm sensitivity calibrates the magnitude of the noise by which $\Breve{\boldsymbol{\beta}}_k^{(m)}$ has to be perturbed to preserve privacy. Unlike the existing privacy-preserving methods where the noise is added to the output of the algorithm \cite{Huang2020,Zhuhan,Dworkbook,Dworkcalibrate,Liu2019,Tao2017}, in D-ZOA, the noise is inherent.

In addition to Assumptions 1 and 2, we introduce the following assumption that is widely used in the literature, see, e.g., \cite{Huang2020,Liu2019,Tao2017}.

\textit{Assumption 3:} There exists a constant $c_1$ such that $\norm{\nabla\ell(\cdot)}\le c_1$  where $\ell(\cdot)$ is the loss function defined in Section~\ref{secttwo}.

Similar to the classical methods of differential privacy analysis, e.g., \cite{Huang2020,Liu2019}, we first define the $l_2$ norm sensitivity. Subsequently, we estimate the $l_2$-norm sensitivity of $\boldsymbol{\beta}_k^{(m)}$.

\begin{defi}
The $l_2$-norm sensitivity of $\Breve{\boldsymbol{\beta}}_k^{(m)}$ is defined as
\begin{equation}
\label{firsteqdeftwo}
\Delta_{k,2}=\max_{\mathcal{D}_k,\mathcal{D}_k'}\norm{\Breve{\boldsymbol{\beta}}_{k,\mathcal{D}_k}^{(m)}-\Breve{\boldsymbol{\beta}}_{k,\mathcal{D}_k'}^{(m)}} 
\end{equation}
where $\Breve{\boldsymbol{\beta}}_{k,\mathcal{D}_k}^{(m)}$ and $\Breve{\boldsymbol{\beta}}_{k,\mathcal{D}_k'}^{(m)}$ denote the local primal variables for two neighboring datasets $\mathcal{D}_k$ and $\mathcal{D}_k'$ differing in only one data sample, i.e., one row of $\mathbf{X}_k$ and the corresponding entry of $\mathbf{y}_k$.\looseness=-1
\end{defi}
The $l_2$-norm sensitivity of $\Breve{\boldsymbol{\beta}}_k^{(m)}$ is an upper bound on $\norm{\Breve{\boldsymbol{\beta}}_{k,\mathcal{D}_k}^{(m)}-\Breve{\boldsymbol{\beta}}_{k,\mathcal{D}_k'}^{(m)}}$ and is computed as in the following lemma.

\begin{lem}
Under Assumption 3, the $l_2$-norm sensitivity of $\Breve{\boldsymbol{\beta}}_k^{(m)}$ is given by
\begin{equation}
\label{eqfirstlemmathree}
\Delta_{k,2}=\frac{c_1}{\rho|\mathcal{N}_k|N_k}.
\end{equation}
\end{lem}
\begin{proof}\renewcommand{\qedsymbol}{}
See Appendix C.
\end{proof}

In the next subsection, we present our main result proving that the proposed D-ZOA algorithm is differentially private.  

\subsection{Intrinsic $(\epsilon,\delta)$-Differential Privacy Guarantee}

In this section, we prove that, at each iteration of Algorithm \ref{alg:IDP-D-ZOA}, $(\epsilon,\delta)$-differential privacy is guaranteed.

\begin{theo}
Let $\epsilon\in\mathopen(0,1\mathclose]$ and 
\begin{equation}
\label{firsteqtheo}
\sigma_{k}=\frac{c_1\sqrt{2.1\log(1.25/\delta)}}{\rho|\mathcal{N}_k|N_k\epsilon}.
\end{equation}
Under Assumption 3, at each iteration of D-ZOA, $(\epsilon,\delta)$-differential privacy is guaranteed. Specifically, for any neighboring datasets $\mathcal{D}_k$ and $\mathcal{D}_k'$ and any output $\boldsymbol{\beta}_k^{(m)}$, the following inequality holds:
\begin{equation}
\label{secondeqtheo}
\text{Pr}[\boldsymbol{\beta}_{k,\mathcal{D}_k}^{(m)}]\le e^{\epsilon}\text{Pr}[\boldsymbol{\beta}_{k,\mathcal{D}_k'}^{(m)}]+\delta.
\end{equation}
\end{theo}
\begin{proof}\renewcommand{\qedsymbol}{}
See Appendix D.
\end{proof}
Theorem 1 also shows that the variance of the inherent noise is inversely proportional to the privacy parameter $\epsilon$. This implies that a higher variance leads to a smaller $\epsilon$ and higher privacy guarantee. In fact, a smaller $\epsilon$ implies that the ratio of the probability distributions of $\boldsymbol{\beta}_{k,\mathcal{D}_k}^{(m)}$ and $\boldsymbol{\beta}_{k,\mathcal{D}_k'}^{(m)}$ is smaller, which means less information is available to a sniffing/spoofing adversary through $\boldsymbol{\beta}_k$ hence the improved privacy \cite{Tao2017}.

The following corollary shows the connection between the privacy parameter $\epsilon$ and the number of samples $J$ in the most private case of the variance $\sigma_{k}$ approaching its upper bound.

\begin{cor}
If $\{\mathbf{g}_{j,k}^{(t)}\}_{j=1}^J$ are i.i.d., $J$ is sufficiently large, and Assumption 3 holds, we have
\begin{equation}
\begin{aligned}
\epsilon&=\frac{c_1}{\rho|\mathcal{N}_k|N_k}\sqrt{2.1JP\log(1.25/\delta)}\\
&\times\left(\frac{cR^2\alpha_{0}^2}{\log(2P)}(s_1(1+\log(P))+s_2)-\frac{4\norm{\boldsymbol{\beta}^c}^2}{T}\right)^{-\frac{1}{2}}.
\label{eqfirstcor}
\end{aligned}
\end{equation}
\end{cor}
\begin{proof}
The proof follows from equating the expression for $\sigma_{k}$ in Theorem 1, \eqref{firsteqtheo}, and the expression for the trace of the covariance of $\boldsymbol{\beta}_k^{(m)}$ in Lemma 1 where the trace of $\boldsymbol{\Psi}_k^{(t)}$ has been replaced by the upper-bound derived in Lemma 2. Solving the resultant equation for $\epsilon$ yields \eqref{eqfirstcor}. 
\end{proof}
Corollary 1 provides the highest possible level of privacy guarantee that can be ensured by D-ZOA due to its inherent randomness brought about by using a zeroth-order method in the inner loop.

\subsection{Total Privacy Leakage}

In this subsection, we consider the total privacy leakage of the proposed D-ZOA algorithm. Since D-ZOA is an $M$-fold adaptive algorithm, we utilize the results of \cite{Abadi} together with the moments accountant method to evaluate its total privacy leakage. The main result is summarized in the following theorem.

\begin{theo}
Let $\epsilon\in\mathopen(0,1\mathclose]$ and 
\begin{equation}
\label{firsteqtheosecond}
\sigma_{k}=\frac{c_1\sqrt{2.1\log(1.25/\delta)}}{\rho|\mathcal{N}_k|N_k\epsilon}.
\end{equation}
Under Assumption 3, Algorithm \ref{alg:IDP-D-ZOA} guarantees $(\bar{\epsilon},\delta)$-differential privacy where 
\begin{equation}
\label{secondeqtheosecond}
\bar{\epsilon}=\epsilon\sqrt{\frac{M\log(1/\delta)}{1.05\log(1.25/\delta)}}.
\end{equation}
\end{theo}
\begin{proof}
The proof is obtained by using the log moments of the privacy loss and their linear composability in the same way as in \cite[Theorem 2]{Huang2020}.
\end{proof}

\section{Convergence Analysis} \label{sectfive}

The convergence of D-ZOA to the centralized solution is established by corroborating that both inner and outer loops of the algorithm converge. The convergence of the inner loop can be verified following \cite[Theorem 2]{Duchi2015}, i.e., it can be shown that, under Assumptions 1 and 2, there exists a constant $c$ such that, for each $T$ representing a fixed number of inner-loop iterations, the following inequality holds:
\begin{equation}
\begin{aligned}
&\mathbb{E}[\mathcal{F}(\hat{\mathbf{w}}^{(T)})-\mathcal{F}(\mathbf{w}^*)]\\
\le&c\frac{RL\sqrt{P}}{\sqrt{T}}\Big(\max\{\alpha_{0},\alpha_{0}^{-1}\}\sqrt{\log(2P)}+\frac{u_1\log(2T)}{\sqrt{T}}\Big)
\end{aligned}
\label{convergence1}
\end{equation}
where
\begin{equation*}
\hat{\mathbf{w}}^{(T)}=\frac{1}{T}\sum_{t=1}^{T}\mathbf{w}^{(t)}.
\end{equation*}
In \cite{Duchi2015}, it is shown that $c=0.5$ is suitable when $\boldsymbol{\nu}_1$ and $\boldsymbol{\nu}_2$ are sampled from a normal distribution. The convergence of the outer loop can be proven by verifying the convergence of a fully distributed ADMM with inexact primal updates. 

To present the convergence result, we construct the auxiliary sequence
\begin{equation*}
\mathbf{r}^{(m)}=\sum_{s=0}^m\mathbf{Q}\mathbf{w}^{(s)}
\end{equation*}
and define the auxiliary vector $\mathbf{q}^{(m)}$ and the auxiliary matrix $\mathbf{G}$ as  
\begin{equation}
\label{convergenceone}
\mathbf{q}^{(m)}=\begin{bmatrix} \mathbf{r}^{(m)} \\ \mathbf{w}^{(m)} \end{bmatrix}, \quad \mathbf{G}=\begin{bmatrix} \rho\mathbf{I}_P & \mathbf{0}_{P\times P}\\ \mathbf{0}_{P\times P} & \rho\frac{\mathbf{L_+}}{2}\end{bmatrix}.
\end{equation}
The convergence results of \cite{Varshney2018}, \cite{Zhuhan}, and \cite{Varshneyarxiv} can now be adapted to D-ZOA as per the following theorem that also provides an explicit privacy-accuracy trade-off.

\begin{theo}
\theoremstyle{plain}
If $f(\cdot)$ is convex, $\{\mathbf{g}_{j,k}^{(t)}\}_{j=1}^J$ are i.i.d., $J$ is sufficiently large, and Assumption 3 holds, for any $M>0$, we have 
\begin{equation}
\begin{aligned}
&\mathbb{E}[f(\hat{\mathbf{w}}^{(M)})-f(\mathbf{w}^*)]\\
&\le\frac{\norm{\mathbf{q}^{(0)}-\mathbf{q}}_{\mathbf{G}}^2}{M}+\frac{2.1c_1^2P\rho\log(1.25/\delta)\lambda_{\text{max}}^2(\mathbf{L}_+)}{2\rho^2|\mathcal{N}_k|^2N_k^2\epsilon^2\lambda_{\text{min}}(\mathbf{L}_-)}
\label{convergence4}
\end{aligned}
\end{equation}
where $\mathbf{q}=[\mathbf{r}^\mathsf{T}, (\mathbf{w}^*)^\mathsf{T}]^\mathsf{T}$ and 
\begin{equation*}
\hat{\mathbf{w}}^{(M)}=\frac{1}{M}\sum_{m=1}^M\Breve{\mathbf{w}}^{(m)}.
\end{equation*}
\end{theo}
\begin{proof}
See Appendix E.
\end{proof}

Theorem 3 reveals a privacy-accuracy trade-off offered by D-ZOA. When the privacy guarantee is stronger (smaller $\epsilon$ and $\delta$), the accuracy is lower. It also shows that D-ZOA converges at a rate of $\mathcal{O}(1/M)$ where $M$ is the number of iterations of the ADMM outer loop.

\section{Simulations} \label{sectsix}

In this section, we present some simulated examples to evaluate the performance and the privacy-accuracy trade-off of the proposed D-ZOA algorithm. We benchmark the performance and the privacy-accuracy trade-off of D-ZOA against an existing differentially-private ADMM-based algorithm, called DP-ADMM and proposed in \cite{Huang2020}. DP-ADMM is suitable for learning problems with non-smooth objective functions. It is a differentially-private algorithm that is not fully-distributed since it needs a central coordinator to average the dual variable and the perturbed primal variable over the network at every iteration. As for the application, we consider a distributed version of the empirical risk minimization problem defined by lasso \cite{Boyd2010}. 

\begin{figure}[t!]
\begin{center}
\includegraphics[trim={1cm 1.25cm 1cm 1.25cm},clip,scale=0.85]{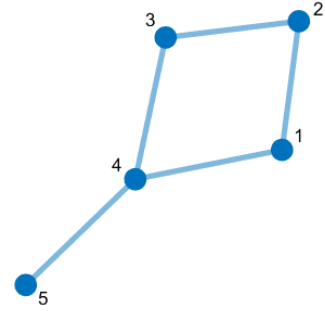}
\caption{Topology of the considered multi-agent network.}
\label{first_fig}
\end{center}
\end{figure}

The network-wide observations are represented by a design matrix $\mathbf{X}\in\mathbb{R}^{N\times P}$ and a response vector $\mathbf{y}\in\mathbb{R}^{N\times 1}$ where $N$ is the number of data samples and $P$ is the number of features in each sample. The matrix $\mathbf{X}$ consists of $K$ submatrices $\mathbf{X}_k$, i.e., $\mathbf{X}=[\mathbf{X}_1^{\mathsf{T}}, \mathbf{X}_2^{\mathsf{T}},\hdots, \mathbf{X}_K^{\mathsf{T}}]^{\mathsf{T}}$, and the vector $\mathbf{y}$ consists of $K$ subvectors $\mathbf{y}_k$, i.e., $\mathbf{y}=\left[\mathbf{y}_1^{\mathsf{T}},\mathbf{y}_2^{\mathsf{T}},\hdots, \mathbf{y}_K^{\mathsf{T}} \right]^{\mathsf{T}}$, as the data is distributed among the agents and each agent $k$ holds its respective $\mathbf{X}_k\in\mathbb{R}^{N_k\times P}$ and $\mathbf{y}_k\in\mathbb{R}^{N_k\times 1}$ where $N=\sum_{k=1}^K N_k$. The parameter vector that establishes a linear regression between $\mathbf{X}$ and $\mathbf{y}$ is $\boldsymbol{\beta}\in\mathbb{R}^{P\times 1}$. In the centralized approach, a lasso estimate of $\boldsymbol{\beta}$ is given by 
\begin{equation}
\boldsymbol{\beta}^c=\arg\min_{\boldsymbol{\beta}}\{\norm{\mathbf{X}\boldsymbol{\beta}-\mathbf{y}}^2+\eta\norm{\boldsymbol{\beta}}_1\}.
\label{sim1}
\end{equation}
 In the distributed setting, we solve problem \eqref{secondeq} with 
\begin{equation}
\begin{aligned}
\sum_{j=1}^{N_k}\ell(\mathbf{x}_{k,j},y_{k,j};\boldsymbol{\beta}_k)&=\norm{\mathbf{X}_k\boldsymbol{\beta}_k-\mathbf{y}_k}^2\\
R(\boldsymbol{\beta}_k)&=\norm{\boldsymbol{\beta}_k}_1.
\label{sim2}
\end{aligned}
\end{equation}

\begin{figure*}[htp]
  \centering
  \subfigure[$\delta=10^{-3}$]{\includegraphics[scale=1]{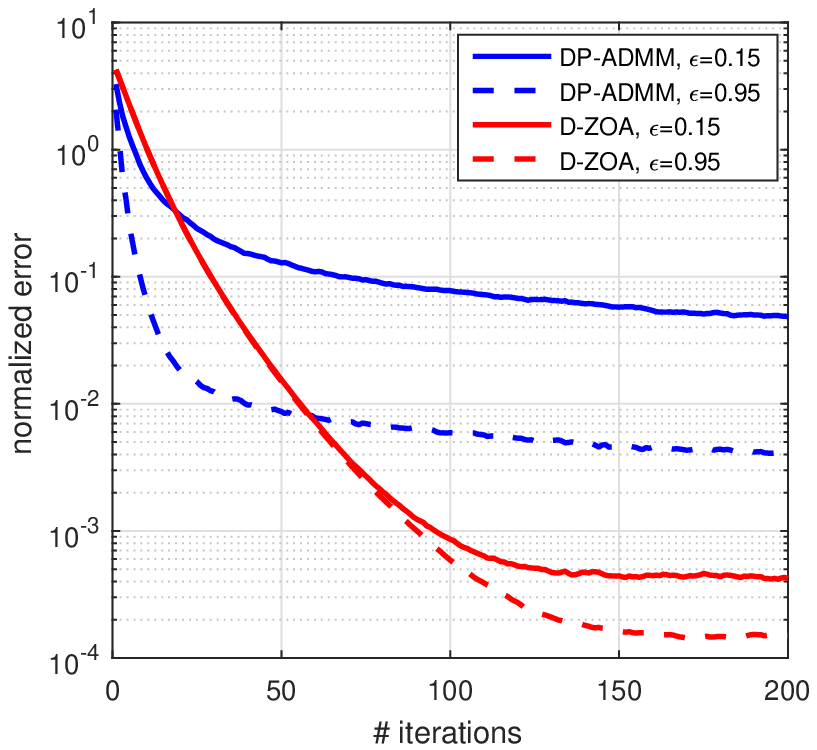}}\quad
  \subfigure[$\delta=10^{-6}$]{\includegraphics[scale=1]{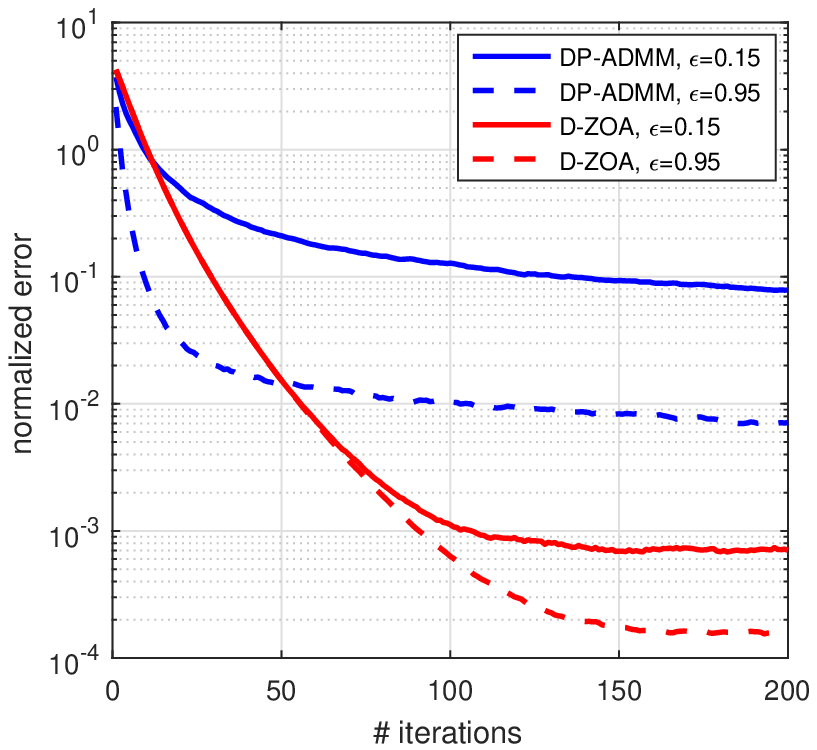}}
\caption{Normalized error of DP-ADMM and D-ZOA for two values of $\epsilon$ and fixed $\delta$.}  
\end{figure*}

\begin{figure*}[htp]
  \centering
  \subfigure[$\delta=10^{-6}$ and $\delta=10^{-3}$]{\includegraphics[scale=1]{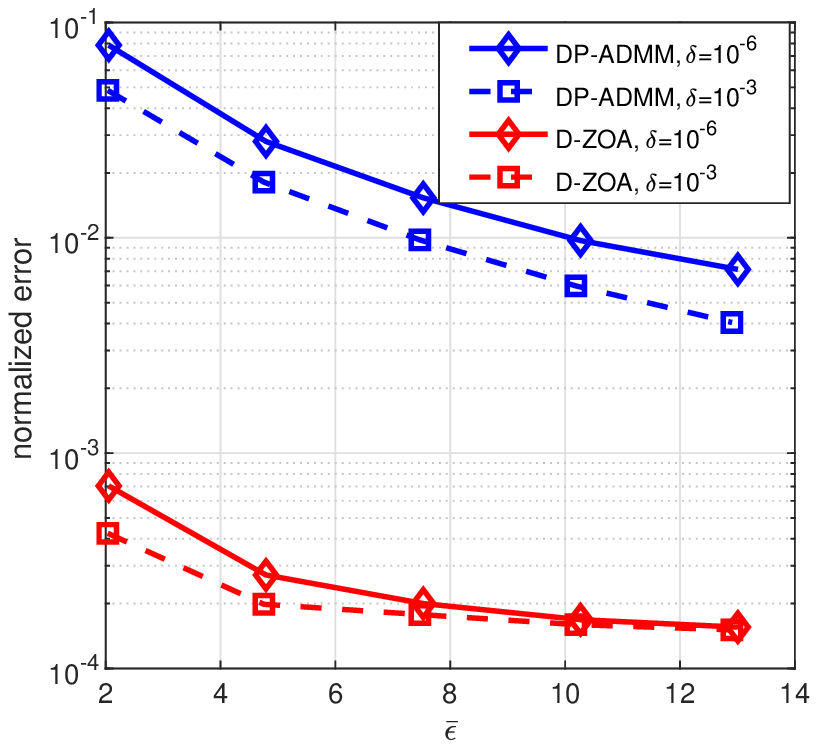}}\quad
  \subfigure[$\epsilon=0.15$ and $\epsilon=0.95$]{\includegraphics[scale=1]{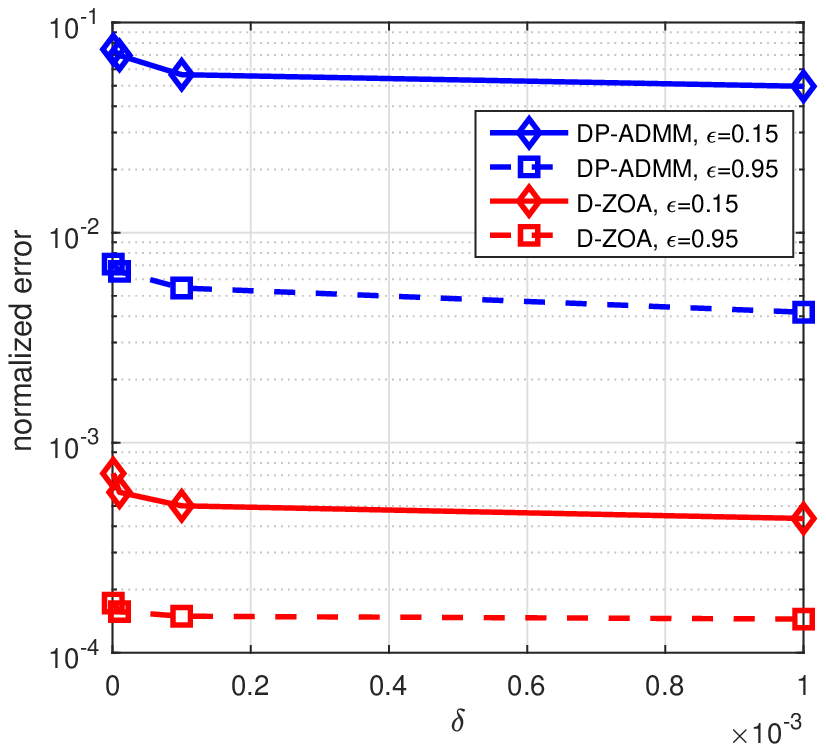}}
\caption{Privacy-accuracy trade-off.}  
\end{figure*}

We assess the performance of the D-ZOA algorithm over a network of $K=5$ agents with the topology as shown in Fig. 1. The number of samples at each agent is set to $N_k=20$ $\forall k\in\mathcal{K}$ and the total number of samples is $N=100$. The number of features in each sample is $P=10$. For each agent $k\in\mathcal{K}$, we create a $2P\times P$ local observation matrix $\mathbf{X}_k$ whose entries are i.i.d. zero-mean unit-variance Gaussian random variables. The response vector $\mathbf{y}$ is synthesized as 
\begin{equation*}
    \mathbf{y}=\mathbf{X}\boldsymbol{\omega}+\boldsymbol{\psi}
\end{equation*}
where $\boldsymbol{\omega}\in\mathbb{R}^{P}$ and $\boldsymbol{\psi}\in\mathbb{R}^{M}$ are random vectors with distributions $\mathcal{N}(\mathbf{0},\mathbf{I}_P)$ and $\mathcal{N}(\mathbf{0},0.1\mathbf{I}_N)$, respectively. The data are preprocessed by normalizing the columns of $\mathbf{X}$ to guarantee that the maximum value of each column is $1$ and by normalizing the rows to enforce their $l_2$-norm to be less than $1$ as in \cite{Huang2020}. This is motivated by the need for homogeneous scaling of the features. Therefore, we have $c_1=1$. The regularization parameter is set to $\eta=1$ and the penalty parameter is set to $\rho=4$. The number of iterations of the ADMM outer loop is set to $200$. For the inner loop, the number of iterations is set to $100$ and the smoothing constant $u_1$ to $1$. To apply the central limit theorem, $J$ needs to be sufficiently large, i.e., $J\geq 30$. We set $J=30$ and calculate $\alpha_0$ according to \eqref{eqfirstcor}. Performance of D-ZOA is evaluated using the normalized error between the centralized solutions $\boldsymbol{\beta}^c$ as per \eqref{sim1} and the local estimates. It is defined as $\sum_{k=1}^{K}{\norm{\boldsymbol{\beta}_{k}-\boldsymbol{\beta}^c}^2}/{\norm{\boldsymbol{\beta}^c}^2}$ where $\boldsymbol{\beta}_{k}$ denotes the local estimate at agent $k$. The centralized solution $\boldsymbol{\beta}^c$ is computed using the convex optimization toolbox CVX \cite{cvx}. Results are obtained by averaging over $100$ independent trials.

In Fig. 2, we plot the normalized error versus the outer loop iteration index for both D-ZOA and DP-ADMM. The plots show that both algorithms converge to the centralized solution for two different values of $\epsilon$ and $\delta$. In all plots, accuracy improves as $\epsilon$ increases. This is consistent with both Theorem 3 and \cite[Theorem 3]{Huang2020}. The hyper-parameters in DP-ADMM are tuned to achieve the best accuracy and convergence rate. The faster convergence of DP-ADMM is mainly due to its extra centralized processing. However, D-ZOA has higher accuracy than DP-ADMM.

In Fig. 3, we illustrate the privacy-accuracy trade-off for both D-ZOA and DP-ADMM. The figures show that D-ZOA and DP-ADMM achieve higher accuracy with larger $\epsilon$ and $\delta$. In Fig. 3(a), we show the normalized error versus the privacy parameter $\bar{\epsilon}$ as given in \eqref{secondeqtheosecond} for $\delta=10^{-6}$ and $\delta=10^{-3}$. We observe that D-ZOA outperforms DP-ADMM in terms of accuracy likely due to its intrinsic privacy-preserving properties. Fig. 3(b) also attests to the superiority of D-ZOA over DP-ADMM when $\epsilon=0.15$ and $\epsilon=0.95$ and $\delta$ varies between $10^{-6}$ and $10^{-2}$.

\section{Conclusion} \label{sectseven}

We proposed an intrinsically privacy-preserving consensus-based algorithm for solving a class of distributed regularized ERM problems where first-order information is hard or even impossible to obtain. We recast the original problem into an equivalent constrained optimization problem whose structure is suitable for distributed implementation via ADMM. We employed a zeroth-order method, known as the two-point stochastic-gradient algorithm, to minimize the augmented Lagrangian in the primal update step. We proved that the inherent randomness due to employing the zeroth-order method makes the D-ZOA algorithm intrinsically privacy-preserving. In addition, we used the moments accountant method to show that the total privacy leakage of D-ZOA grows sublinearly with the number of ADMM iterations. We verified the convergence of D-ZOA to the optimal solution as well as studying its privacy-preserving properties through both theoretical analysis and numerical simulations.

\appendices
\section{Proof of Lemma 1}
\begin{proof}
We prove this lemma in two steps. First, we prove that $\mathbb{E}[\boldsymbol{\beta}_k^{(m)}]=\Breve{\boldsymbol{\beta}}_k^{(m)}$. Second, we calculate the covariance of $\boldsymbol{\beta}_k^{(m)}$.

We prove that $\mathbb{E}[\boldsymbol{\beta}_k^{(m)}]=\Breve{\boldsymbol{\beta}}_k^{(m)}$ by induction over $m$.\\
\textit{Base case:} Since $\boldsymbol{\beta}_k^{(0)}=\Breve{\boldsymbol{\beta}}_k^{(0)}=\mathbf{0}$, we have $\mathbb{E}[\boldsymbol{\beta}_k^{(0)}]=\Breve{\boldsymbol{\beta}}_k^{(0)}$.\\
\textit{Induction step:} We assume that $\mathbb{E}[\boldsymbol{\beta}_k^{(m-1)}]=\Breve{\boldsymbol{\beta}}_k^{(m-1)}$ as the induction hypothesis. Considering \eqref{fourthprivacya} and \eqref{firstprivacy}, we have
\begin{equation}
\begin{aligned}
\label{eqsecondproofzero}
\mathbb{E}[\boldsymbol{\beta}_k^{(m)}]=\ &\mathbb{E}[\Breve{\boldsymbol{\beta}}_k^{(m)}]+\mathbb{E}[\boldsymbol{\xi}_k^{(m)}]\\
=\ &\Breve{\boldsymbol{\beta}}_k^{(m)}+\mathbb{E}[\boldsymbol{\xi}_k^{(m)}]\\
=\ &-\frac{1}{2\rho|\mathcal{N}_k|}\nabla f_k(\Breve{\boldsymbol{\beta}}_k^{(m)})+\frac{1}{2|\mathcal{N}_k|}\Bigl(|\mathcal{N}_k|\Breve{\boldsymbol{\beta}}_k^{(m-1)}\nonumber\\&+\sum_{l\in\mathcal{N}_k}\Breve{\boldsymbol{\beta}}_l^{(m-1)}\Bigr)-\frac{1}{2\rho|\mathcal{N}_k|}\boldsymbol{\gamma}_k^{(m-1)}+\mathbb{E}[\boldsymbol{\xi}_k^{(m)}]\\
=\ &-\frac{1}{2\rho|\mathcal{N}_k|}\mathbb{E}[\nabla f_k(\Breve{\boldsymbol{\beta}}_k^{(m)})]\\
&+\frac{1}{2|\mathcal{N}_k|}\Bigl(|\mathcal{N}_k|\mathbb{E}[\boldsymbol{\beta}_k^{(m-1)}]+\sum_{l\in\mathcal{N}_k}\mathbb{E}[\boldsymbol{\beta}_l^{(m-1)}]\Bigr)\\
&-\frac{1}{2\rho|\mathcal{N}_k|}\mathbb{E}[\boldsymbol{\gamma}_k^{(m-1)}]+\mathbb{E}[\boldsymbol{\xi}_k^{(m)}]\\
=\ &\mathbb{E}[\boldsymbol{\beta}_k^{(m)}]+\mathbb{E}[\boldsymbol{\xi}_k^{(m)}],
\end{aligned}
\end{equation}
which implies that $\mathbb{E}[\boldsymbol{\xi}_k^{(m)}]=\mathbf{0}$. Therefore, $\mathbb{E}[\boldsymbol{\beta}_k^{(m)}]=\Breve{\boldsymbol{\beta}}_k^{(m)}$.

Since the sequence $\{\mathbf{g}_{j,k}^{(t)}\}_{j=1}^J$ is i.i.d. and $J$ is sufficiently large, thanks to the central limit theorem \cite{Papoulis}, we have
\begin{equation*}
\mathbf{g}_k^{(t)} = \frac{1}{J}\sum_{j=1}^J\mathbf{g}_{j,k}^{(t)}\sim\mathcal{N}\Bigr(\boldsymbol{\mu}_k^{(t)},\frac{1}{J}\boldsymbol{\Psi}_k^{(t)}\Bigl).
\end{equation*}
In view of the additive property of the normal distribution and recalling that $\boldsymbol{\beta}_k^{(m)}=-\sum_{t=1}^{T}\alpha_{t}\mathbf{g}_k^{(t)}$, we have
\begin{equation}
\label{eqsecondproofone}
\text{cov}[\boldsymbol{\beta}_k^{(m)}]=\frac{1}{J}\sum_{t=1}^{T}\alpha_t^2\Psi_k^{(t)}.
\end{equation}
We also know $\mathbb{E}[\boldsymbol{\beta}_k^{(m)}]=\Breve{\boldsymbol{\beta}}_k^{(m)}$. Therefore, $\boldsymbol{\beta}_k^{(m)}$ is distributed as 
\begin{equation}
\label{eqthirdproofone}
\boldsymbol{\beta}_k^{(m)}\sim\mathcal{N}\Bigl(\Breve{\boldsymbol{\beta}}_k^{(m)},\frac{1}{J}\sum_{t=1}^{T}\alpha_t^2\Psi_k^{(t)}\Bigr).
\end{equation}
\end{proof}

\section{Proof of Lemma 2}
\begin{proof}
It is easy to verify that 
\begin{equation}
\begin{aligned}
\text{tr}(\boldsymbol{\Psi}_k^{(t)})&=\text{tr}(\mathbb{E}[\mathbf{g}_{j,k}^{(t)}(\mathbf{g}_{j,k}^{(t)})^\mathsf{T}]-\mathbb{E}[\mathbf{g}_{j,k}^{(t)}]\mathbb{E}[\mathbf{g}_{j,k}^{(t)}]^\mathsf{T})\\
&=\text{tr}(\mathbb{E}[\mathbf{g}_{j,k}^{(t)}(\mathbf{g}_{j,k}^{(t)})^\mathsf{T}])-\norm{\mathbb{E}[\mathbf{g}_{j,k}^{(t)}]}^2\\
&=\mathbb{E}[\text{tr}(\mathbf{g}_{j,k}^{(t)}(\mathbf{g}_{j,k}^{(t)})^\mathsf{T})]-\norm{\boldsymbol{\mu}_k^{(t)}}^2\\&=\mathbb{E}\left[\norm{\mathbf{g}_{j,k}^{(t)}}^2\right]-\norm{\boldsymbol{\mu}_k^{(t)}}^2.
\label{eqfirstprooftwo}
\end{aligned}
\end{equation}
By \cite[Lemma 2]{Duchi2015}, there exists a constant $c$ such that  
\begin{equation}
\label{eqsecondprooftwo}
\mathbb{E}\left[\norm{\mathbf{g}_{j,k}^{(t)}}^2\right]\le cL^2P\Bigl(\sqrt{\frac{u_{2,t}}{u_{1,t}}}P+1+\log(P)\Bigr).
\end{equation}
Since $u_{2,t}/u_{1,t}=P^{-2}t^{-1}$, we have 
\begin{equation}
\label{eqthirdprooftwo}
\mathbb{E}\left[\norm{\mathbf{g}_{j,k}^{(t)}}^2\right]\le cL^2P\Bigl(\frac{1}{\sqrt{t}}+1+\log(P)\Bigr).
\end{equation}
In addition, from $\boldsymbol{\beta}_k^{(m)}=-\sum_{t=1}^{T}\alpha_{t}\mathbf{g}_k^{(t)}$ and \eqref{firstprivacy}, we have 
\begin{equation}
\label{eqfourthprooftwo}
\boldsymbol{\Breve{\beta}}_k^{(m)}=-\sum_{t=1}^{T}\alpha_{t}\boldsymbol{\mu}_k^{(t)}.
\end{equation}
Taking the Euclidean norm of both sides in \eqref{eqfourthprooftwo} and using the triangle inequality, we have
\begin{equation}
\begin{aligned}
\norm{\boldsymbol{\Breve{\beta}}_k^{(m)}}=\norm{-\sum_{t=1}^{T}\alpha_{t}\boldsymbol{\mu}_k^{(t)}}\le\sum_{t=1}^{T}|\alpha_{t}|\norm{\boldsymbol{\mu}_k^{(t)}}.
\label{eqfifthprooftwo}
\end{aligned}
\end{equation}
Squaring both sides of \eqref{eqfifthprooftwo} and using the Cauchy-Schwarz inequality, we get
\begin{equation}
\begin{aligned}
\norm{\boldsymbol{\Breve{\beta}}_k^{(m)}}^2&\le\Bigl(\sum_{t=1}^{T}|\alpha_{t}|\norm{\boldsymbol{\mu}_k^{(t)}}\Bigr)^2\\
&\le T\sum_{t=1}^{T}|\alpha_{t}|^2\norm{\boldsymbol{\mu}_k^{(t)}}^2
\label{eqsixthprooftwo}
\end{aligned}
\end{equation}
and consequently
\begin{equation}
\label{eqseventhprooftwo}
-\frac{1}{JP}\sum_{t=1}^{T}\alpha_{t}^2\norm{\boldsymbol{\mu}_k^{(t)}}^2\le-\frac{1}{TJP}\norm{\boldsymbol{\Breve{\beta}}_k^{(m)}}^2.
\end{equation}
Using \eqref{eqthirdprooftwo}, \eqref{eqseventhprooftwo}, and the definition of $\alpha_t$ after \eqref{thirdzeroth}, we have
\begin{equation}
\begin{aligned}
&\frac{1}{JP}\sum_{t=1}^{T-1}\alpha_{t}^2\text{tr}(\boldsymbol{\Psi}_k^{(t)})\\
\le&\frac{1}{JP}\sum_{t=1}^{T-1}\alpha_t^2
cL^2P\Bigl(\frac{1}{\sqrt{t}}+1+\log(P)\Bigr)\\&-
\frac{1}{JP}\sum_{t=1}^{T-1}\alpha_t^2\norm{\boldsymbol{\mu}_k^{(t)}}^2\\
=&\frac{1}{JP}\frac{c\alpha_0^2R^2}{\log(2P)}\Bigl(\sum_{t=1}^{T}\frac{1}{t\sqrt{t}}+(1+\log(P))\sum_{t=1}^{T}\frac{1}{t}\Bigr)\\&-\frac{1}{TJP}\norm{\boldsymbol{\Breve{\beta}}_k^{(m)}}^2.
\label{eqeigthprooftwo}
\end{aligned}
\end{equation}
Defining $s_1=\sum_{t=1}^{T-1}t^{-1}$ and $s_2=\sum_{t=1}^{T-1}t^{-1.5}$, \eqref{eqeigthprooftwo} simplifies to 
\begin{equation}
\begin{aligned}
&\frac{1}{JP}\sum_{t=1}^{T}\alpha_t^2\text{tr}(\boldsymbol{\Psi}_k^{(t)})\\
&\le\frac{c\alpha_0^2R^2}{JP\log(2P)}\Bigl(s_1(1+\log(P))+s_2\Bigr)-\frac{\norm{\Breve{\boldsymbol{\beta}}_k^{(m)}}^2}{TJP}.
\label{eqninthprooftwo}
\end{aligned}
\end{equation}

Considering that the algorithm converges as proven in Section \ref{sectfive}, i.e., $\Breve{\boldsymbol{\beta}}_k^{(m)}\rightarrow\boldsymbol{\beta}^c$ as $m\rightarrow\infty$, $\Breve{\boldsymbol{\beta}}_k^{(0)}=\mathbf{0}$, and the triangle inequality, for $m>0$ we have
\begin{equation}
\begin{aligned}
\label{eqtenthprooftwo}
\left|\norm{\Breve{\boldsymbol{\beta}}_k^{(m)}}-\norm{\boldsymbol{\beta}^c}\right|\le\norm{\Breve{\boldsymbol{\beta}}_k^{(m)}-\boldsymbol{\beta}^c}\le\norm{\boldsymbol{\beta}^c},
\end{aligned}
\end{equation}
which implies $\norm{\Breve{\boldsymbol{\beta}}_k^{(m)}}\le 2\norm{\boldsymbol{\beta}^c}$. 
Therefore, we obtain
\begin{equation}
\begin{aligned}
&\frac{1}{JP}\sum_{t=1}^{T}\alpha_t^2\text{tr}(\boldsymbol{\Psi}_k^{(t)})\\
&\le\frac{c\alpha_0^2R^2}{JP\log(2P)}\Bigl(s_1(1+\log(P))+s_2\Bigr)-\frac{4\norm{\boldsymbol{\beta}^c}^2}{TJP}.
\label{eqeleventhprooftwo}
\end{aligned}\qedhere
\end{equation}
\end{proof}

\section{Proof of Lemma 3}
\begin{proof}
From the adopted exact primal update equation \eqref{fourthprivacya}, we obtain
\begin{equation}
\begin{aligned}
\Breve{\boldsymbol{\beta}}_{k,\mathcal{D}_k}^{(m)}=&-\frac{0.5}{\rho|\mathcal{N}_k|}\Bigl(\frac{1}{N_k}\sum_{j=1}^{N_k}\nabla\ell(\mathbf{x}_{k,j},y_{k,j};\Breve{\boldsymbol{\beta}}_k)+\boldsymbol{\gamma}_k^{(m-1)}\Bigr)\\
&+\frac{0.5}{|\mathcal{N}_k|}\Bigl(\Breve{\boldsymbol{\beta}}_k^{(m-1)}+\sum_{l\in\mathcal{N}_k}\Breve{\boldsymbol{\beta}}_l^{(m-1)}+\frac{\eta \nabla R(\Breve{\boldsymbol{\beta}}_k)}{\rho K}\Bigr)\\
\Breve{\boldsymbol{\beta}}_{k,\mathcal{D}_k'}^{(m)}=&-\frac{0.5}{\rho|\mathcal{N}_k|}\Bigl(\frac{1}{N_k}\sum_{j=1}^{N_k-1}\nabla\ell(\mathbf{x}_{k,j},y_{k,j};\Breve{\boldsymbol{\beta}}_k)+\boldsymbol{\gamma}_k^{(m-1)}\\&+\frac{1}{N_k}\nabla\ell(\mathbf{x}_{k,N_k}',y_{k,N_k}';\Breve{\boldsymbol{\beta}}_k)  \Bigr)\\
&+\frac{0.5}{|\mathcal{N}_k|}\Bigl(\Breve{\boldsymbol{\beta}}_k^{(m-1)}+\sum_{l\in\mathcal{N}_k}\Breve{\boldsymbol{\beta}}_l^{(m-1)}+\frac{\eta \nabla R(\Breve{\boldsymbol{\beta}}_k)}{\rho K}\Bigr).
\label{eqfirstprooflemmathree}
\end{aligned}
\end{equation}
Using Assumption 3, the quantity $\norm{\Breve{\boldsymbol{\beta}}_{k,\mathcal{D}_k}^{(m)}-\Breve{\boldsymbol{\beta}}_{k,\mathcal{D}_k'}^{(m)}}$ is upper bounded as follows
\begin{equation}
\begin{aligned}
&\norm{\Breve{\boldsymbol{\beta}}_{k,\mathcal{D}_k}^{(m)}-\Breve{\boldsymbol{\beta}}_{k,\mathcal{D}_k'}^{(m)}}\\
&=\frac{\norm{\nabla\ell(\mathbf{x}_{k,N_k}',y_{k,N_k}';\Breve{\boldsymbol{\beta}}_k)-\nabla\ell(\mathbf{x}_{k,N_k},y_{k,N_k};\Breve{\boldsymbol{\beta}}_k)}}{2\rho|\mathcal{N}_k|N_k}\\
&\le\frac{c_1}{\rho|\mathcal{N}_k|N_k}.
\label{eqsecondprooflemmathree}
\end{aligned}
\end{equation}
\end{proof}

\section{Proof of Theorem 1}
\begin{proof}

The privacy loss due to sharing $\boldsymbol{\beta}_k^{(m)}$ is calculated as 
\begin{equation}
\label{eqfirstproof}
\left|\log\frac{\text{Pr}[\boldsymbol{\beta}_{k,\mathcal{D}_k}^{(m)}]}{\text{Pr}[\boldsymbol{\beta}_{k,\mathcal{D}_k'}^{(m)}]}\right|=\left|\log\frac{\text{Pr}[\boldsymbol{\xi}_{k,\mathcal{D}_k}^{(m)}]}{\text{Pr}[\boldsymbol{\xi}_{k,\mathcal{D}_k'}^{(m)}]}\right|
\end{equation}
where the equality holds since the Jacobian matrix of the linear transformation from $\boldsymbol{\beta}_k^{(m)}$ to $\boldsymbol{\xi}_k^{(m)}$ is the identity matrix. Furthermore, as the entries of $\boldsymbol{\xi}_k^{(m)}$, denoted by $\xi_{s,k}^{(m)}$, are independent of each other, for any entry $s$, we have

\begin{equation}
\label{eqsecondproof}
\left|\log\frac{\text{Pr}[\boldsymbol{\xi}_{k,\mathcal{D}_k}^{(m)}]}{\text{Pr}[\boldsymbol{\xi}_{k,\mathcal{D}_k'}^{(m)}]}\right|=\left|\log\frac{\text{Pr}[{\xi}_{s,k,\mathcal{D}_k}^{(m)}]}{\text{Pr}[{\xi}_{s,k,\mathcal{D}_k'}^{(m)}]}\right|.
\end{equation}
Hence, we have
\begin{equation}
\begin{aligned}
&\left|\log\frac{\text{Pr}[\boldsymbol{\beta}_{k,\mathcal{D}_k}^{(m)}]}{\text{Pr}[\boldsymbol{\beta}_{k,\mathcal{D}_k'}^{(m)}]}\right|=\left|\log\frac{\text{Pr}[{\xi}_{s,k,\mathcal{D}_k}^{(m)}]}{\text{Pr}[{\xi}_{s,k,\mathcal{D}_k'}^{(m)}]}\right|\\
=&\left|\log\frac{\exp\Bigl(-\frac{1}{2\sigma_k^2}[\xi_{s,k,\mathcal{D}_k}^{(m)}]^2\Bigr)}{\exp\Bigl(-\frac{1}{2\sigma_k^2}[\xi_{s,k,\mathcal{D}_k}^{(m)}+(\Breve{\beta}_{s,k,\mathcal{D}_k}^{(m)}-\Breve{\beta}_{s,k,\mathcal{D}_k'}^{(m)})]^2\Bigr)}\right|.
\label{eqthirdproof}
\end{aligned}
\end{equation}
Via the triangle inequality, \eqref{eqthirdproof} leads to
\begin{equation}
\begin{aligned}
\left|\log\frac{\text{Pr}[\boldsymbol{\beta}_{k,\mathcal{D}_k}^{(m)}]}{\text{Pr}[\boldsymbol{\beta}_{k,\mathcal{D}_k'}^{(m)}]}\right|\le&\frac{1}{2\sigma_{k}^2}\left|2{\xi}_{s,k}^{(m)}(\Breve{\beta}_{s,k,\mathcal{D}_k}^{(m)}-\Breve{\beta}_{s,k,\mathcal{D}_k'}^{(m)})\right|\\
&+\frac{1}{2\sigma_{k}^2}(\Breve{\beta}_{s,k,\mathcal{D}_k}^{(m)}-\Breve{\beta}_{s,k,\mathcal{D}_k'}^{(m)})^2.
\label{eqthirdproofb}
\end{aligned}
\end{equation}
Since $\norm{\nabla\ell(\cdot)}\le c_1$, using Lemma 3, we have
\begin{equation}
\begin{aligned}
\left|\Breve{\beta}_{s,k,\mathcal{D}_k}^{(m)}-\Breve{\beta}_{s,k,\mathcal{D}_k'}^{(m)}\right|
&<\norm{\Breve{\boldsymbol{\beta}}_{k,\mathcal{D}_k}^{(m)}-\Breve{\boldsymbol{\beta}}_{k,\mathcal{D}_k'}^{(m)}}\\
&\le\frac{c_1}{\rho|\mathcal{N}_k|N_k}.
\label{eqfourthproof}
\end{aligned}
\end{equation}
Hence, substituting $\sigma_{k}$ in \eqref{firsteqtheo} into \eqref{eqthirdproofb}, we obtain 
\begin{equation}
\begin{aligned}
\left|\log\frac{\text{Pr}[\boldsymbol{\beta}_{k,\mathcal{D}_k}^{(m)}]}{\text{Pr}[\boldsymbol{\beta}_{k,\mathcal{D}_k'}^{(m)}]}\right|
\le\frac{\rho|\mathcal{N}_k|N_k\epsilon^2}{2.1c_1\log(1.25/\delta)}\left|\xi_{s,k}^{(m)}+\frac{c_1}{2\rho|\mathcal{N}_k|N_k}\right|.
\label{eqfifthproof}
\end{aligned}
\end{equation}

When
\begin{equation*}
|\xi_{s,k}^{(m)}|\le \frac{c_1}{\rho|\mathcal{N}_k|N_k}\left(2.1\epsilon^{-1}\log(1.25/\delta)-0.5\right),
\end{equation*}
the privacy loss is bounded by $\epsilon$. Hence, let us define
\begin{equation*}
r=\frac{c_1}{\rho|\mathcal{N}_k|N_k}\left(2.1\epsilon^{-1}\log(1.25/\delta)-0.5\right).
\end{equation*}
Subsequently, we need to prove that
\begin{equation*}
\text{Pr}[|\xi_{s,k}^{(m)}|>r]\le\delta
\end{equation*}
or equivalently

\begin{equation*}
\text{Pr}[\xi_{s,k}^{(m)}>r]\le 0.5\delta.
\end{equation*}
Using the tail bound of the normal distribution $\mathcal{N}(0,\sigma_{k}^2)$ \cite{Dworkbook}, we obtain 
\begin{equation}
\label{eqsixthproof}
\text{Pr}[\xi_{s,k}^{(m)}>r]\le\frac{\sigma_{k}}{r\sqrt{2\pi}}\exp\Bigl(-\frac{r^2}{2\sigma_{k}^2}\Bigr).
\end{equation}
Since $\delta$ is assumed to be small $(\le 0.01)$ and $\epsilon\le 1$, we have 
\begin{equation*}
\frac{\sigma_{k}}{r}<1 \quad \text{and} \quad -\frac{r^2}{2\sigma_{k}^2}<\log(0.5\sqrt{2\pi}\delta).
\end{equation*}
Therefore,
\begin{equation*}
\text{Pr}[\xi_{s,k}^{(m)}>r]<0.5\delta,
\end{equation*}
which implies
\begin{equation*}
\text{Pr}[|\xi_{s,k}^{(m)}|>r]\le\delta.
\end{equation*}
By defining
\begin{equation*}
\begin{aligned}
\mathbb{A}_1&=\{\xi_{s,k}^{(m)}:|\xi_{s,k}^{(m)}|\le r\}\\
\mathbb{A}_2&=\{\xi_{s,k}^{(m)}:|\xi_{s,k}^{(m)}|> r\},
\end{aligned}
\end{equation*}
we have
\begin{equation}
\begin{aligned}
\text{Pr}[\boldsymbol{\beta}_{k,\mathcal{D}_k}^{(m)}]
&=\text{Pr}[\Breve{\beta}_{s,k,\mathcal{D}_k}^{(m)}+\xi_{s,k}^{(m)}:\xi_{s,k}^{(m)}\in\mathbb{A}_1]\\
&+\text{Pr}[\Breve{\beta}_{s,k,\mathcal{D}_k}^{(m)}+\xi_{s,k}^{(m)}:\xi_{s,k}^{(m)}\in\mathbb{A}_2]\\
&<e^{\epsilon}\text{Pr}[\boldsymbol{\beta}_{k,\mathcal{D}_k'}^{(m)}]+\delta,
\label{eqseventhproof}
\end{aligned}
\end{equation}
which concludes the proof by showing that, at each iteration of D-ZOA, $(\epsilon,\delta)$-differential privacy is guaranteed.
\end{proof}

\section{Proof of Theorem 3}
\begin{proof}
Using the first-order condition for convexity \cite{BoydStephenP2004Co}, we have
\begin{equation}
\label{convergenceoneproof}
f(\Breve{\mathbf{w}}^{(m)})-f(\mathbf{w}^*)\le(\Breve{\mathbf{w}}^{(m)}-\mathbf{w}^*)^\mathsf{T}\nabla f(\Breve{\mathbf{w}}^{(m)}).
\end{equation}
In addition, in virtue of \cite[Lemma 1 and Lemma 2]{Varshney2018}, $\Breve{\mathbf{w}}^{(m)}$ satisfies the following equation
\begin{equation}
\label{convergencetwoproof}
\frac{\nabla f(\Breve{\mathbf{w}}^{(m)})}{\rho}=2\mathbf{H}\boldsymbol{\xi}^{(m)}-2\mathbf{Q}\mathbf{r}^{(m)}-\mathbf{L}_+(\mathbf{w}^{(m)}-\mathbf{w}^{(m-1)}).
\end{equation}
Therefore, by using \eqref{convergenceoneproof}, \eqref{convergencetwoproof}, \cite[Lemma 3, Lemma 4 and Lemma 5]{Varshney2018}, and the steps in the proof of \cite[Theorem 1]{Varshneyarxiv}, we can show that, for any $\mathbf{r}\in\mathbb{R}^{KP}$ and $m>0$, we have
\begin{equation}
\begin{aligned}
&\frac{f(\Breve{\mathbf{w}}^{(m)})-f(\mathbf{w}^*)}{\rho}+2\mathbf{r}^\mathsf{T}\mathbf{Q}\Breve{\mathbf{w}}^{(m)}\\
\le&(\Breve{\mathbf{w}}^{(m)}-\mathbf{w}^*)^\mathsf{T}\left(-\mathbf{L}_+(\Breve{\mathbf{w}}^{(m)}-\Breve{\mathbf{w}}^{(m-1)})\right.\\
&-\mathbf{L}_+(\Breve{\mathbf{w}}^{(m-1)}-\mathbf{w}^{(m-1)})-2\mathbf{Q}(\mathbf{r}^{(m)}-\mathbf{r})\\
&\left.+\mathbf{L}_-(\mathbf{w}^{(m)}-\Breve{\mathbf{w}}^{(m)})\right)\\
=&\frac{\norm{\mathbf{q}^{(m-1)}-\mathbf{q}}_{\mathbf{G}}^2}{\rho}-\frac{\norm{\mathbf{q}^{(m)}-\mathbf{q}}_{\mathbf{G}}^2}{\rho}-\frac{\norm{\mathbf{q}^{(m)}-\mathbf{q}^{(m-1)}}_{\mathbf{G}}^2}{\rho}\\
&+(\Breve{\mathbf{w}}^{(m)}-\mathbf{w}^*)^\mathsf{T}\mathbf{L}_+(\mathbf{w}^{(m-1)}-\Breve{\mathbf{w}}^{(m-1)})\\
&+
(\Breve{\mathbf{w}}^{(m)}-\mathbf{w}^*)^\mathsf{T}\mathbf{L}_-(\mathbf{w}^{(m)}-\Breve{\mathbf{w}}^{(m)})+2(\boldsymbol{\xi}^{(m)})^\mathsf{T}\mathbf{Q}(\mathbf{r}^{(m)}-\mathbf{r})\\
=&\frac{\norm{\mathbf{q}^{(m-1)}-\mathbf{q}}_{\mathbf{G}}^2}{\rho}-\frac{\norm{\mathbf{q}^{(m)}-\mathbf{q}}_{\mathbf{G}}^2}{\rho}\\
&-\norm{\mathbf{Q}\Breve{\mathbf{w}}^{(m)}}^2-\norm{\mathbf{Q}\boldsymbol{\xi}^{(m)}}^2+2(\boldsymbol{\xi}^{(m)})^\mathsf{T}\mathbf{Q}(\mathbf{r}^{(m)}-\mathbf{r})\\
&+2\Bigl(\frac{\mathbf{L}_+}{2}(\Breve{\mathbf{w}}^{(m)}-\mathbf{w}^*)\Bigr)^\mathsf{T}(\mathbf{w}^{(m-1)}-\Breve{\mathbf{w}}^{(m-1)})
\label{convergencethreeproof}
\end{aligned}
\end{equation}
where $\mathbf{q}=[\mathbf{r}^\mathsf{T}, (\mathbf{w}^*)^\mathsf{T}]^\mathsf{T}$.

For any symmetric matrix $\mathbf{X}\in\mathbb{R}^{P\times P}$ and vector $\mathbf{y}\in\mathbb{R}^{P}$, we have
\begin{equation*}
\norm{\mathbf{y}}^2\lambda_{\text{min}}(\mathbf{X})\le\mathbf{y}^\mathsf{T}\mathbf{X}\mathbf{y}\le\norm{\mathbf{y}}^2\lambda_{\text{max}}(\mathbf{X})
\end{equation*}
and, for any $\mathbf{a}, \mathbf{b}\in\mathbb{R}^P$ and $\tau\in\mathbb{R}_+$, we have \begin{equation*}
2\mathbf{a}^\mathsf{T}\mathbf{b}\le\tau^{-1}\norm{\mathbf{a}}^2+\tau\norm{\mathbf{b}}^2.
\end{equation*}
Therefore, \eqref{convergencethreeproof} yields
\begin{equation}
\begin{aligned}
&\frac{f(\Breve{\mathbf{w}}^{(m)})-f(\mathbf{w}^*)}{\rho}+2\mathbf{r}^\mathsf{T}\mathbf{Q}\Breve{\mathbf{w}}^{(m)}\\
\le&\frac{\norm{\mathbf{q}^{(m-1)}-\mathbf{q}}_{\mathbf{G}}^2}{\rho}-\frac{\norm{\mathbf{q}^{(m)}-\mathbf{q}}_{\mathbf{G}}^2}{\rho}-\norm{\mathbf{Q}\boldsymbol{\xi}^{(m)}}^2\\
&-\frac{\lambda_{\text{min}}(\mathbf{L}_-)}{2}\norm{\Breve{\mathbf{w}}^{(m)}-\mathbf{w}^*}^2+\frac{1}{\tau}\norm{\frac{\mathbf{L}_+}{2}(\Breve{\mathbf{w}}^{(m)}-\mathbf{w}^*)}^2\\
&+\tau\norm{\mathbf{w}^{(m-1)}-\Breve{\mathbf{w}}^{(m-1)}}^2+2(\boldsymbol{\xi}^{(m)})^\mathsf{T}\mathbf{Q}(\mathbf{r}^{(m)}-\mathbf{r}).
\label{convergencefourproof}
\end{aligned}
\end{equation}
By setting
\begin{equation*}
\tau=\frac{\lambda_{\text{max}}^2(\mathbf{L}_+)}{2\lambda_{\text{min}}(\mathbf{L}_-)},
\end{equation*}
\eqref{convergencefourproof} leads to 
\begin{equation}
\begin{aligned}
&\frac{f(\Breve{\mathbf{w}}^{(m)})-f(\mathbf{w}^*)}{\rho}+2\mathbf{r}^\mathsf{T}\mathbf{Q}\Breve{\mathbf{w}}^{(m)}\\
\le&\frac{\norm{\mathbf{q}^{(m-1)}-\mathbf{q}}_{\mathbf{G}}^2}{\rho}-\frac{\norm{\mathbf{q}^{(m)}-\mathbf{q}}_{\mathbf{G}}^2}{\rho}-\norm{\mathbf{Q}\boldsymbol{\xi}^{(m)}}^2\\
&+\frac{\lambda_{\text{max}}^2(\mathbf{L}_+)}{2\lambda_{\text{min}}(\mathbf{L}_-)}\norm{\mathbf{w}^{(m-1)}-\Breve{\mathbf{w}}^{(m-1)}}^2\\
&+2(\boldsymbol{\xi}^{(m)})^\mathsf{T}\mathbf{Q}(\mathbf{r}^{(m)}-\mathbf{r})\\
\le&\frac{\norm{\mathbf{q}^{(m-1)}-\mathbf{q}}_{\mathbf{G}}^2}{\rho}-\frac{\norm{\mathbf{q}^{(m)}-\mathbf{q}}_{\mathbf{G}}^2}{\rho}\\&+\frac{\lambda_{\text{max}}^2(\mathbf{L}_+)}{2\lambda_{\text{min}}(\mathbf{L}_-)}\norm{\boldsymbol{\xi}^{(m-1)}}^2+2(\boldsymbol{\xi}^{(m)})^\mathsf{T}\mathbf{Q}(\mathbf{r}^{(m)}-\mathbf{r}).
\label{convergencefiveproof}
\end{aligned}
\end{equation}

Setting $\mathbf{r}=\mathbf{0}_{P}$ and summing both sides of \eqref{convergencefiveproof} over $m=1$ to $M$ gives
\begin{equation}
\begin{aligned}
&\frac{1}{\rho}\sum_{m=1}^M(f(\Breve{\mathbf{w}}^{(m)})-f(\mathbf{w}^*))\le\frac{1}{\rho}\norm{\mathbf{q}^{(0)}-\mathbf{q}}_{\mathbf{G}}^2\\
&+\sum_{m=1}^M\frac{\lambda_{\text{max}}^2(\mathbf{L}_+)}{2\lambda_{\text{min}}(\mathbf{L}_-)}\norm{\boldsymbol{\xi}^{(m-1)}}^2+2(\boldsymbol{\xi}^{(m)})^\mathsf{T}\mathbf{Q}\mathbf{r}^{(m)}.
\label{convergencesixproof}
\end{aligned}
\end{equation}
Using Jensen's inequality \cite{Papoulis}, \eqref{eqfirstlemmatwo}, \eqref{firsteqtheo}, and applying the expectation operator to both sides of \eqref{convergencesixproof}, we obtain
\begin{equation}
\begin{aligned}
&\mathbb{E}[f(\hat{\mathbf{w}}^{(M)})-f(\mathbf{w}^*)]\\
\le&\frac{1}{M}\norm{\mathbf{q}^{(0)}-\mathbf{q}}_{\mathbf{G}}^2+\frac{\rho\lambda_{\text{max}}^2(\mathbf{L}_+)}{2M\lambda_{\text{min}}(\mathbf{L}_-)}\sum_{m=1}^M\mathbb{E}\left[\norm{\boldsymbol{\xi}^{(m-1)}}^2\right]\\
\le&\frac{\norm{\mathbf{q}^{(0)}-\mathbf{q}}_{\mathbf{G}}^2}{M}+\frac{\rho\lambda_{\text{max}}^2(\mathbf{L}_+)\sum_{m=1}^M\text{tr}\left(\text{cov}[\boldsymbol{\xi}^{(m-1)}]\right)}{2M\lambda_{\text{min}}(\mathbf{L}_-)}\\
\le&\frac{\norm{\mathbf{q}^{(0)}-\mathbf{q}}_{\mathbf{G}}^2}{M}+\frac{\rho\lambda_{\text{max}}^2(\mathbf{L}_+)}{2M\lambda_{\text{min}}(\mathbf{L}_-)}\sum_{m=1}^M P\sigma_k^2\\
=&\frac{\norm{\mathbf{q}^{(0)}-\mathbf{q}}_{\mathbf{G}}^2}{M}+\frac{2.1c_1^2P\rho\log(1.25/\delta)\lambda_{\text{max}}^2(\mathbf{L}_+)}{2\rho^2|\mathcal{N}_k|^2N_k^2\epsilon^2\lambda_{\text{min}}(\mathbf{L}_-)}
\label{convergencesevenproof}
\end{aligned}
\end{equation}
where
\begin{equation*}
\hat{\mathbf{w}}^{(M)}=\frac{1}{M}\sum_{m=1}^M\Breve{\mathbf{w}}^{(m)}.\qedhere
\end{equation*}

\end{proof}

\ifCLASSOPTIONcaptionsoff
  \newpage
\fi

\bibliographystyle{IEEEtran}
\bibliography{IEEEabrv,references}

\end{document}